\documentclass[draftclsnofoot,onecolumn,11pt]{IEEEtran}

\usepackage{amsmath}
\usepackage{amssymb}
\usepackage{amsthm}
\usepackage{dsfont}
\usepackage{graphicx,multirow}
\usepackage{cases,cite,colortbl}
\newtheorem{thm} {Theorem}
\newtheorem{lem} {Lemma}

\newtheorem{assumption} {Assumption}

\begin{document}
\title{Diffusion Strategies Outperform Consensus Strategies for Distributed Estimation over Adaptive Networks}

\author{Sheng-Yuan Tu, \IEEEmembership{Student Member,~IEEE}
        and Ali H. Sayed, \IEEEmembership{Fellow,~IEEE}
\thanks{Copyright (c) 2012 IEEE. Personal use of this material is
permitted. However, permission to use this material for any other
purposes must be obtained from the IEEE by sending a request to
pubs-permissions@ieee.org.}
\thanks{This work was supported in part
by NSF grants CCF-1011918 and CCF-0942936. An earlier version of
this work appeared in \cite{Tu11c}. The authors are with the
Department of Electrical Engineering, University of California, Los
Angeles (e-mail: shinetu@ee.ucla.edu; sayed@ee.ucla.edu).}}

\maketitle
\begin{abstract}
Adaptive networks consist of a collection of nodes with adaptation
and learning abilities. The nodes interact with each other on a
local level and diffuse information across the network to solve
estimation and inference tasks in a distributed manner. In this
work, we compare the mean-square performance of two main strategies
for distributed estimation over networks: consensus strategies and
diffusion strategies. The analysis in the paper confirms that under
constant step-sizes, diffusion strategies allow information to
diffuse more thoroughly through the network and this property has a
favorable effect on the evolution of the network: diffusion networks
are shown to converge faster and reach lower mean-square deviation
than consensus networks, and their mean-square stability is
insensitive to the choice of the combination weights. In contrast,
and surprisingly, it is shown that consensus networks can become
unstable even if all the individual nodes are stable and able to
solve the estimation task on their own. When this occurs,
cooperation over the network leads to a catastrophic failure of the
estimation task. This phenomenon does not occur for diffusion
networks: we show that stability of the individual nodes always
ensures stability of the diffusion network irrespective of the
combination topology. Simulation results support the theoretical
findings.
\end{abstract}


\section{Introduction}
Adaptive networks consist of a collection of spatially distributed
nodes that are linked together through a topology and that cooperate
with each other through local interactions. Adaptive networks are
well-suited to perform decentralized information processing and
inference tasks \cite{Lopes08,Cattivelli10} and to model complex and
self-organized behavior encountered in biological systems
\cite{Tu11,Cattivelli11}.

We examine two types of fully decentralized strategies, namely,
consensus strategies and diffusion strategies. The consensus
strategy was originally proposed in the statistics literature
\cite{Degroot74} and has since then been developed into an elegant
procedure to enforce agreement among cooperating nodes. Average
consensus and gossip algorithms have been studied extensively in
recent years, especially in the control literature
\cite{Xiao04,Boyd06,Olfati04,Aysal09,Sardellitti10,Jakovetic10}, and
applied to the study of multi-agent formations
\cite{Jadbabaie03,Ren05}, distributed optimization
\cite{Tsitsiklis86,Nedic09}, and distributed estimation problems
\cite{Dimakis10,Kar11,Schizas09}. Original implementations of the
consensus strategy relied on the use of two time-scales
\cite{Xiao06,Barbarossa07,Johansson08}: one time-scale for the
collection of measurements across the nodes and another time-scale
to iterate sufficiently enough over the collected data to attain
agreement before the process is repeated. Unfortunately, two
time-scale implementations hinder the ability to perform real-time
recursive estimation and adaptation when measurement data keep
streaming in. For this reason, in this work, we focus instead on
consensus implementations that operate in a single time-scale. Such
implementations appear in several recent works, including
\cite{Nedic09,Schizas09,Dimakis10,Kar11}, and are largely motivated
by the procedure developed earlier in
\cite{Tsitsiklis86,Bertsekas97} for the solution of distributed
optimization problems.

The second class of algorithms that we consider deals with diffusion
strategies, which were originally introduced for the solution of
distributed estimation and adaptation problems in
\cite{Lopes06,Sayed07,Cattivelli08b,Lopes08,Cattivelli10}. The main
motivation for the introduction of diffusion strategies in these
works was the desire to develop distributed schemes that are able to
respond in real-time to \emph{continuous streaming} of data at the
nodes by operating over a \emph{single} time-scale. A useful
overview of diffusion strategies appears in \cite{Sayed13}. Since
their inception, diffusion strategies have been applied to model
various forms of complex behavior encountered in nature
\cite{Tu11,Cattivelli11}; they have also been adopted to solve
distributed optimization problems advantageously in
\cite{Ram10,Srivastava11,Chen11}; and have been studied under varied
conditions in \cite{Li09,Takahashi10b,Chouvardas11,Abdolee11} as
well. Diffusion strategies are inherently single time-scale
implementations and are therefore naturally amenable to real-time
and recursive implementations. It turns out that the dynamics of the
consensus and diffusion strategies differ in important ways, which
in turn impact the mean-square behavior of the respective networks
in a fundamental manner.

The analysis in this paper will confirm that under constant
step-sizes, diffusion strategies allow information to diffuse more
thoroughly through networks and this property has a favorable effect
on the evolution of the network. It will be shown that diffusion
networks converge faster and reach lower mean-square deviation than
consensus networks, and their mean-square stability is insensitive
to the choice of the combination weights. In comparison, and
surprisingly, it is shown that consensus networks can become
unstable even if all the individual nodes are stable and able to
solve estimation task on their own. In other words, the learning
curve of a cooperative consensus network can diverge even if the
learning curves for the non-cooperative individual nodes converge.
When this occurs, cooperation over the network leads to a
catastrophic failure of the estimation task. This behavior does not
occur for diffusion networks: we will show that stability of the
individual nodes is sufficient to ensure stability of the diffusion
network \emph{regardless} of the combination weights. The properties
revealed in this paper indicate that there needs to be some care
with the use of consensus strategies for adaptation because they can
lead to network failure even if the individual nodes are stable and
well-behaved. The analysis also suggests that diffusion strategies
provide a proper way to enforce cooperation over networks; their
operation is such that diffusion networks will always remain stable
irrespective of the combination topology.

\section{Estimation Strategies over Networks}
Consider a network consisting of $N$ nodes distributed over a
spatial domain. Two nodes are said to be neighbors if they can
exchange information. The neighborhood of node $k$ is denoted by
$\mathcal{N}_k$. The nodes in the network would like to estimate an
unknown $M\times 1$ vector, $w^\circ$. At every time instant, $i$,
each node $k$ is able to observe realizations $\{d_k(i),u_{k,i}\}$
of a scalar random process $\boldsymbol{d}_{k}(i)$ and a $1\times M$
vector random process $\boldsymbol{u}_{k,i}$ with a
positive-definite covariance matrix,
$R_{u,k}=\mathbb{E}\boldsymbol{u}_{k,i}^*\boldsymbol{u}_{k,i}>0$,
where $\mathbb{E}$ denotes the expectation operator. All vectors in
our treatment are column vectors with the exception of the
regression vector, $\boldsymbol{u}_{k,i}$, which is taken to be a
row vector for convenience of presentation. The random processes
$\{\boldsymbol{d}_{k}(i),\boldsymbol{u}_{k,i}\}$ are related to
$w^\circ$ via the linear regression model \cite{Sayed08}:
\begin{equation} \label{eq2}
    \boldsymbol{d}_k(i) = \boldsymbol{u}_{k,i}w^\circ+\boldsymbol{v}_k(i)
\end{equation}
where $\boldsymbol{v}_k(i)$ is measurement noise with variance
$\sigma^2_{v,k}$ and assumed to be temporally white and spatially
independent, i.e.,
\begin{equation}
    \mathbb{E}\boldsymbol{v}^*_k(i)\boldsymbol{v}_l(j)=
    \sigma^2_{v,k}\cdot\delta_{kl}\cdot\delta_{ij}
\end{equation}
in terms of the Kronecker delta function. The regression data
$\boldsymbol{u}_{k,i}$ are likewise assumed to be temporally white
and spatially independent. The noise $\boldsymbol{v}_k(i)$ and the
regressors $\{\boldsymbol{u}_{l,j}\}$ are assumed to be independent
of each other for all $\{k,l,i,j\}$. All random processes are
assumed to be zero mean. Note that we use boldface letters to denote
random quantities and normal letters to denote their realizations or
deterministic quantities. Models of the form (\ref{eq2}) are useful
in capturing many situations of interest, such as estimating the
parameters of some underlying physical phenomenon, tracking a moving
target by a collection of nodes, or estimating the location of a
nutrient source or predator in biological networks (see, e.g.,
\cite{Tu11,Cattivelli11,Sayed08}); these models are also useful in
the study of the performance limits of combinations of adaptive
filters \cite{Arenas05,Candido10,Kozat10,Xia11b}.

The objective of the network is to estimate $w^\circ$ in a
distributed manner through an online learning process. The nodes
estimate $w^\circ$ by seeking to minimize the following global cost
function:
\begin{equation} \label{eq1}
    J^{\text{glob}}(w)\triangleq\sum_{k=1}^{N}
    \mathbb{E}|\boldsymbol{d}_k(i)-\boldsymbol{u}_{k,i}w|^2.
\end{equation}
In the sequel, we describe the algorithms pertaining to the
consensus and diffusion strategies that we study in this article, in
addition to the non-cooperative mode of operation. Afterwards, we
move on to the main theme of this work, which is to show why
diffusion networks outperform consensus networks. We may remark that
the same strategies can be used to optimize global cost functions
where the individual costs are not necessarily quadratic in $w$ as
in (\ref{eq1}). Most of the mean-square analysis performed here can
be extended to this more general scenario --- see, e.g.,
\cite{Chen11,Towfic12} and the references therein.

\subsection{Non-Cooperative Strategy}
In the non-cooperative mode of operation, each node $k$ operates
independently of the other nodes and estimates $w^\circ$ by means of
a local LMS adaptive filter applied to its data
$\{d_k(i),u_{k,i}\}$. The filter update takes the following form
\cite{Haykin02,Sayed08}:
\begin{equation}\label{eq56}\text{(non-cooperative strategy)}\quad \boxed{
    w_{k,i}=w_{k,i-1}+\mu_ku_{k,i}^*[d_k(i)-u_{k,i}w_{k,i-1}]}
\end{equation}
where $\mu_k>0$ is the \emph{constant} step-size used by node $k$.
In (\ref{eq56}), the vector $w_{k,i}$ denotes the estimate for
$w^\circ$ that is computed by node $k$ at time $i$. Note that for
the underlying model where $R_{u,k}>0$ for all $k$, every individual
node can employ (\ref{eq56}) to estimate $w^\circ$ independently if
desired. Studies allowing for other observability conditions for
diffusion and consensus strategies, including possibly singular
covariance matrices, appear in \cite{Kar11,Abdolee12}.

\subsection{Cooperative Strategies}
In the cooperative mode of operation, nodes interact with their
neighbors by sharing information. In this article, we study three
cooperative strategies for distributed estimation.

\subsubsection*{\underline{B.1. Consensus Strategy}}
The consensus strategy often appears in the literature in the
following form (see, e.g., Eq. (1.20) in \cite{Nedic09}, Eq. (19) in
\cite{Dimakis10}, and Eq. (9) in \cite{Kar11}):
\begin{equation}\label{eq10}
\begin{aligned}
    w_{k,i} = w_{k,i-1}-
    \mu_k(i)\cdot \sum_{l\in\mathcal{N}_k\setminus\{k\}}
    b_{l,k}(w_{k,i-1}-w_{l,i-1})
    +\mu_k(i)\cdot u_{k,i}^*[d_k(i)-u_{k,i}w_{k,i-1}]
\end{aligned}
\end{equation}
where $\{b_{l,k}\}$ is a set of nonnegative coefficients. It should
be noted that in most works on consensus implementations, especially
in the context of distributed optimization problems
\cite{Nedic09,Bertsekas97,Dimakis10,Kar11,Ram10}, the step-sizes
$\{\mu_k(i)\}$ that are used in (\ref{eq10}) depend on the
time-index $i$ and are required to satisfy
\begin{equation}\label{eq21}
    \sum_{i=0}^{\infty}\mu_k(i)=\infty \text{ and }
    \sum_{i=0}^{\infty} \mu_k^2(i)<\infty.
\end{equation}
In other words, for each node $k$, the step-size sequence $\mu_k(i)$
is required to vanish as $i\rightarrow\infty$. Under such
conditions, it is known that consensus strategies allow the nodes to
reach agreement about $w^\circ$
\cite{Nedic09,Kar11,Braca10,Bajovic11}. Here, instead, we will use
\emph{constant} step-sizes $\{\mu_k\}$. This is because we are
interested in studying the adaptation and learning abilities of the
networks. Constant step-sizes are critical to endow networks with
\emph{continuous} adaptation and tracking abilities; otherwise,
under (\ref{eq21}), once the step-sizes have decayed to zero, the
network stops adapting and learning is turned off.

We can rewrite recursion (\ref{eq10}) in a more compact and
revealing form by combining the first two terms on the right-hand
side of (\ref{eq10}) and by introducing the following coefficients:
\begin{equation}\label{eq20}
    a_{l,k}\triangleq\begin{cases}
    1-\sum_{j\in\mathcal{N}_k\setminus\{k\}}\mu_kb_{j,k}, &\text{if
    $l=k$}\\
    \mu_kb_{l,k}, &\text{if $l\in\mathcal{N}_k\setminus\{k\}$}\\
    0, &\text{otherwise}
    \end{cases}
\end{equation}
In this way, recursion (\ref{eq10}) can be rewritten equivalently as
(see, e.g., expression (7.1) in \cite{Bertsekas97} and expression
(1.20) in \cite{Nedic09}):
\begin{equation}\label{eq7}\text{(consensus strategy)}\quad\boxed{
    w_{k,i}=\sum_{l\in\mathcal{N}_k}a_{l,k}w_{l,i-1}
    +\mu_ku_{k,i}^*[d_k(i)-u_{k,i}w_{k,i-1}]}
\end{equation}
The entry $a_{l,k}$ denotes the weight that node $k$ assigns to the
estimate $w_{l,i-1}$ received from its neighbor $l$ (see Fig.
\ref{Fig_2}); note that the weights $\{a_{l,k}\}$ are nonnegative
for $l\neq k$ and that $a_{k,k}$ is nonnegative for sufficiently
small step-sizes. If we collect the nonnegative weights
$\{a_{l,k}\}$ into an $N\times N$ matrix $A$, then it follows from
(\ref{eq20}) that the combination matrix $A$ satisfies the following
properties:
\begin{equation} \label{eq50}\boxed{
    a_{l,k}\geq 0\text{, }A^T\mathds{1}=\mathds{1} \text{, and }
    a_{l,k}=0 \text{ if $l\notin\mathcal{N}_k$}}
\end{equation}
where $\mathds{1}$ is a vector of size $N$ with all entries equal to
one. That is, the weights on the links arriving at node $k$ add up
to one, which is equivalent to saying that the matrix $A$ is
left-stochastic. Moreover, if two nodes $l$ and $k$ are not linked,
then their corresponding entry $a_{l,k}$ is zero.

\begin{figure}
\centering
\includegraphics[width=18em]{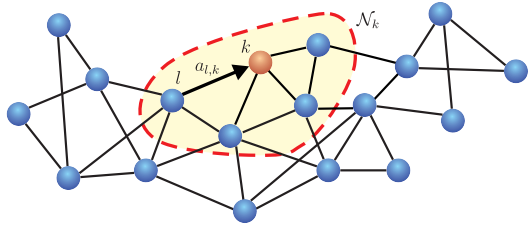}
\caption{A connected network showing the neighborhood of node $k$,
denoted by $\mathcal{N}_k$. The weight $a_{l,k}$ scales the data
transmitted from node $l$ to node $k$ over the edge linking them.}
\label{Fig_2}
\end{figure}

\subsubsection*{\underline{B.2. ATC Diffusion Strategy}}
Diffusion strategies for the optimization of (\ref{eq1}) in a fully
decentralized manner were derived in
\cite{Lopes06,Sayed07,Lopes08,Cattivelli08b,Cattivelli10,Chen11} by
applying a completion-of-squares argument, followed by a stochastic
approximation step and an incremental approximation step --- see
\cite{Sayed13}. The adapt-then-combine (ATC) form of the diffusion
strategy is described by the following update equations
\cite{Cattivelli10}:
\begin{equation} {\label{eq49}}\text{(ATC diffusion strategy)}\quad\boxed{
\begin{aligned}
    \psi_{k,i} &= w_{k,i-1}+\mu_ku_{k,i}^*[d_k(i)-u_{k,i}w_{k,i-1}]\\
    w_{k,i} &= \sum_{l\in\mathcal{N}_k}a_{l,k}\psi_{l,i}
\end{aligned}}.
\end{equation}
The above strategy consists of two steps. The first step of
(\ref{eq49}) involves local adaptation, where node $k$ uses its own
data $\{d_k(i),u_{k,i}\}$ to update its weight estimate from
$w_{k,i-1}$ to an intermediate value $\psi_{k,i}$. The second step
of (\ref{eq49}) is a consultation (combination) step where the
intermediate estimates $\{\psi_{l,i}\}$ from the neighborhood of
node $k$ are combined through weights $\{a_{l,k}\}$ that satisfy
(\ref{eq50}) to obtain the updated weight estimate $w_{k,i}$.

\subsubsection*{\underline{B.3. CTA Diffusion Strategy}}
Another variant of the diffusion strategy is the combine-then-adapt
(CTA) form, which is described by the following update equations
\cite{Lopes08}:
\begin{equation}\label{eq52}\text{(CTA diffusion strategy)}\quad\boxed{
\begin{aligned}
    \psi_{k,i-1} &= \sum_{l\in\mathcal{N}_k}a_{l,k}w_{l,i-1}\\
    w_{k,i} &=\psi_{k,i-1}+\mu_ku_{k,i}^*[d_k(i)-u_{k,i}\psi_{k,i-1}]
\end{aligned}}.
\end{equation}
Thus, comparing the ATC and CTA strategies, we note that the order
of the consultation and adaptation steps are simply reversed. The
first step of (\ref{eq52}) involves a consultation step, where the
existing estimates $\{w_{l,i-1}\}$ from the neighbors of node $k$
are combined through the weights $\{a_{l,k}\}$. The second step of
(\ref{eq52}) is a local adaptation step, where node $k$ uses its own
data $\{d_k(i),u_{k,i}\}$ to update its weight estimate from the
intermediate value $\psi_{k,i-1}$ to $w_{k,i}$.

\subsubsection*{\underline{B.4. Comparing Diffusion and Consensus Strategies}}
For ease of comparison, we rewrite below the recursions that
correspond to the consensus (\ref{eq7}), ATC diffusion (\ref{eq49}),
and CTA diffusion (\ref{eq52}) strategies in a single update:
\begin{align}
    \text{(consensus)}\quad
    w_{k,i}&=\sum_{l\in\mathcal{N}_k}a_{l,k}w_{l,i-1}
    +\mu_ku_{k,i}^*[d_k(i)-u_{k,i}w_{k,i-1}]\label{eq115}\\
    \text{(ATC diffusion)}\quad w_{k,i}&=\sum_{l\in\mathcal{N}_k}a_{l,k} \left(w_{l,i-1}
    + \mu_l u_{l,i}^*[d_l(i)-u_{l,i}w_{l,i-1}]\right)\label{eq53} \\
    \text{(CTA diffusion)}\quad w_{k,i}&=\sum_{l\in\mathcal{N}_k}a_{l,k} w_{l,i-1}
    + \mu_ku_{k,i}^*\left[d_k(i)-u_{k,i}\left(
    \sum_{l\in\mathcal{N}_k}a_{l,k}
    w_{l,i-1}\right)\right].\label{eq54}
\end{align}
Note that the first terms on the right hand side of these recursions
are all the same. For the second terms, only variable $w_{k,i-1}$
appears in the consensus strategy (\ref{eq115}), while the diffusion
strategies (\ref{eq53})-(\ref{eq54}) incorporate the estimates
$\{w_{l,i-1}\}$ from the neighborhood of node $k$ into the update of
$w_{k,i}$. Moreover, in contrast to the consensus (\ref{eq115}) and
CTA diffusion (\ref{eq54}) strategies, the ATC diffusion strategy
(\ref{eq53}) further incorporates the influence of the data
$\{d_{l}(i),u_{l,i}\}$ from the neighborhood of node $k$ into the
update of $w_{k,i}$. These differences in the order by which the
computations are performed have important implications on the
evolution of the weight-error vectors across consensus and diffusion
networks. It is important to note that the diffusion strategies
(\ref{eq53})-(\ref{eq54}) are able to incorporate additional
information into their processing steps \emph{without} being more
complex than the consensus strategy. All three strategies have the
\emph{same} computational complexity and require sharing the same
amount of data (see Table I), as can be ascertained by comparing the
actual implementations (\ref{eq7}), (\ref{eq49}), and (\ref{eq52}).
The key fact to note is that the diffusion implementations first
generate an \emph{intermediate state} variable, which is
subsequently used in the final update. This important ordering of
the calculations has a critical influence on the performance of the
algorithms, as we now move on to reveal.

\begin{table}
\centering \caption{{\rm Comparison of the number of complex
multiplications and additions per iteration, as well as the number
of $M\times 1$ vectors that are exchanged for each iteration of the
algorithms at every node $k$. In the table, the symbol $n_k$ denotes
the degree of node $k$, i.e., the size of its neighborhood
$\mathcal{N}_k$. Observe that all three strategies have
\emph{exactly} the same computational complexity.}}
\begin{tabular}{|c|c|c|c|}
\hline & \textbf{ATC diffusion} (\ref{eq49}) & \textbf{CTA
diffusion} (\ref{eq52}) & \textbf{Consensus} (\ref{eq7}) \\
\hline\hline \textbf{Multiplications} &
    $(n_k+2)M$ & $(n_k+2)M$ & $(n_k+2)M$\\
\hline \textbf{Additions} &
    $(n_k+1)M$ & $(n_k+1)M$ & $(n_k+1)M$\\
\hline \textbf{Vector exchanges} &
    $n_k$ & $n_k$ & $n_k$\\
\hline
\end{tabular}
\end{table}

\section{Mean-Square Performance Analysis}
The mean-square performance of diffusion networks has been studied
in detail in \cite{Lopes08,Cattivelli10,Sayed13} by applying energy
conservation arguments \cite{Sayed08,Naffouri03}. Following
\cite{Cattivelli10}, we will first show how to carry out the
performance analysis in a unified manner that covers both diffusion
and consensus strategies (see Table II further ahead, which
highlights how the parameters for both strategies differ).
Subsequently, we use the resulting performance expressions to carry
out detailed comparisons and to establish and highlight some
surprising and interesting differences in performance.

\subsection{Network Error Recursion}
Let the error vector for an arbitrary node $k$ be denoted by
\begin{equation}
    \tilde{\boldsymbol{w}}_{k,i}\triangleq
    w^\circ-\boldsymbol{w}_{k,i}.
\end{equation}
We collect all error vectors and step-sizes across the network into
a block vector and block matrix:
\begin{align}
    \tilde{\boldsymbol{w}}_{i}&\triangleq
    \text{col}\left\{\tilde{\boldsymbol{w}}_{1,i},\tilde{\boldsymbol{w}}_{2,i},\cdots,
    \tilde{\boldsymbol{w}}_{N,i}\right\}\\
    \mathcal{M}&\triangleq\text{diag}\{\mu_1I_M,\mu_2I_M,\cdots,\mu_NI_M\} \label{eq84}
\end{align}
where the notation $\text{col}\{\cdot\}$ denotes the vector that is
obtained by stacking its arguments on top of each other, and the
notation $\text{diag}\{\cdot\}$ constructs a diagonal matrix from
its arguments. We further introduce the extended combination matrix:
\begin{equation}
    \mathcal{A}\triangleq A\otimes I_M
\end{equation}
where the symbol $\otimes$ denotes the Kronecker product of two
matrices. This construction replaces each entry $a_{l,k}$ in $A$ by
the $M\times M$ diagonal matrix $a_{l,k}I_M$ in $\mathcal{A}$. Then,
if we start from (\ref{eq115}), (\ref{eq53}), or (\ref{eq54}), and
use model (\ref{eq2}), some straightforward algebra similar to
\cite{Cattivelli10,Sayed13} shows that the global error vector
$\tilde{\boldsymbol{w}}_{i}$ for the various strategies evolves
according to the following recursion:
\begin{equation}\label{eq12}\boxed{
    \tilde{\boldsymbol{w}}_{i}=\boldsymbol{\mathcal{B}}_i\cdot
    \tilde{\boldsymbol{w}}_{i-1}-\boldsymbol{y}_i}
\end{equation}
where the quantities $\boldsymbol{\mathcal{B}}_i$ and
$\boldsymbol{y}_i$ are listed in Table II and where
$\boldsymbol{\mathcal{R}}_i$ is a block diagonal matrix and
$\boldsymbol{s}_i$ is a block column vector:
\begin{align}
    \boldsymbol{\mathcal{R}}_i&\triangleq\text{diag}\{\boldsymbol{u}^*_{1,i}\boldsymbol{u}_{1,i},
    \boldsymbol{u}^*_{2,i}\boldsymbol{u}_{2,i},\cdots,
    \boldsymbol{u}^*_{N,i}\boldsymbol{u}_{N,i}\}\\
    \boldsymbol{s}_i&\triangleq\text{col}\{\boldsymbol{u}^*_{1,i}\boldsymbol{v}_{1,i},
    \boldsymbol{u}^*_{2,i}\boldsymbol{v}_{2,i},\cdots,
    \boldsymbol{u}^*_{N,i}\boldsymbol{v}_{N,i}\}.
\end{align}
The coefficient matrix $\boldsymbol{\mathcal{B}}_i$ is an $N\times
N$ block matrix with blocks of size $M\times M$ each. Likewise, the
driving vector $\boldsymbol{y}_i$ is an $N\times 1$ block vector
with entries that are $M\times 1$ each. The matrix
$\boldsymbol{\mathcal{B}}_i$ controls the evolution of the network
error vector $\tilde{\boldsymbol{w}}_{i}$. It is obvious from Table
II that this matrix is different for each of the strategies under
consideration. We shall verify in the sequel that the differences
have critical ramifications when we compare consensus and diffusion
strategies. Note in passing that any of these three distributed
strategies degenerates to the non-cooperative strategy (\ref{eq56})
when $A=I_N$.

\begin{table}
\centering \caption{{\rm The network weight error vector evolves
according to the recursion}
$\tilde{\boldsymbol{w}}_{i}=\boldsymbol{\mathcal{B}}_i\cdot
    \tilde{\boldsymbol{w}}_{i-1}-\boldsymbol{y}_i$,
{\rm where the variables} $\{\boldsymbol{\mathcal{B}}_i,
\boldsymbol{y}_i\}$, {\rm and their respective means or covariances,
are listed below for three cooperative strategies and the
non-cooperative strategy.} }
\renewcommand{\arraystretch}{1.2}
\begin{tabular}{|c|c|c|c|c|}
\hline & \textbf{ATC diffusion} (\ref{eq49}) & \textbf{CTA
diffusion} (\ref{eq52}) & \textbf{Consensus} (\ref{eq7})
& \textbf{Non-cooperative} (\ref{eq56}) \\
\hline\hline $\boldsymbol{\mathcal{B}}_i$ &
    $\mathcal{A}^{T}(I_{NM}-\mathcal{M}\boldsymbol{\mathcal{R}}_i)$ &
    $(I_{NM}-\mathcal{M}\boldsymbol{\mathcal{R}}_i)\mathcal{A}^{T}$ &
    $\mathcal{A}^T-\mathcal{M}\boldsymbol{\mathcal{R}}_i$ &
    $I_{NM}-\mathcal{M}\boldsymbol{\mathcal{R}}_i$\\
\hline $\mathcal{B}\triangleq\mathbb{E}\boldsymbol{\mathcal{B}}_{i}$
&   $\mathcal{A}^T(I_{NM}-\mathcal{M}\mathcal{R})$ &
    $(I_{NM}-\mathcal{M}\mathcal{R})\mathcal{A}^T$ &
    $\mathcal{A}^T-\mathcal{M}\mathcal{R}$ &
    $I_{NM}-\mathcal{M}\mathcal{R}$\\
\hline $\boldsymbol{y}_i$ &
    $\mathcal{A}^T\mathcal{M}\boldsymbol{s}_i$ &
    $\mathcal{M}\boldsymbol{s}_i$ &
    $\mathcal{M}\boldsymbol{s}_i$ &
    $\mathcal{M}\boldsymbol{s}_i$\\
\hline $\mathcal{Y}\triangleq
\mathbb{E}\boldsymbol{y}_i\boldsymbol{y}_i^*$ &
    $\mathcal{A}^T\mathcal{M}\mathcal{S}\mathcal{M}\mathcal{A}$ &
    $\mathcal{M}\mathcal{S}\mathcal{M}$ &
    $\mathcal{M}\mathcal{S}\mathcal{M}$ &
    $\mathcal{M}\mathcal{S}\mathcal{M}$ \\
\hline
\end{tabular}
\end{table}

\subsection{Mean Stability}
We start our analysis by examining the stability in the mean of the
networks, i.e., the stability of the recursion for
$\mathbb{E}\tilde{\boldsymbol{w}}_{i}$. Thus, note that the matrices
$\{\boldsymbol{\mathcal{B}}_{i}\}$ in Table II are random matrices
due to the randomness of the regressors $\{\boldsymbol{u}_{k,i}\}$
in $\boldsymbol{\mathcal{R}}_i$. In other words, the evolution of
the networks is stochastic in nature. Now, since the regressors
$\{\boldsymbol{u}_{k,i}\}$ are temporally white and spatially
independent, then the $\{\boldsymbol{\mathcal{B}}_{i}\}$ are
independent of $\tilde{\boldsymbol{w}}_{i-1}$ for any of the
strategies. Moreover, since the
$\{\boldsymbol{u}_{k,i},\boldsymbol{v}_k(i)\}$ are independent of
each other, then the $\{\boldsymbol{y}_{i}\}$ are zero mean. Taking
expectation of both sides of (\ref{eq12}), we find that the mean of
$\tilde{\boldsymbol{w}}_{i}$ evolves in time according to the
recursion:
\begin{equation}\label{eq5}\boxed{
    \mathbb{E}\tilde{\boldsymbol{w}}_{i}=\mathcal{B}\cdot
    \mathbb{E}\tilde{\boldsymbol{w}}_{i-1}}
\end{equation}
where $\mathcal{B}\triangleq\mathbb{E}\boldsymbol{\mathcal{B}}_{i}$
is shown in Table II and
\begin{align}\label{eq120}
    \mathcal{R}\triangleq \mathbb{E}\boldsymbol{\mathcal{R}}_{i}
    =\text{diag}\{R_{u,1},R_{u,2},\cdots,R_{u,N}\}.
\end{align}
The necessary and sufficient condition to ensure mean stability of
the network (namely,
$\mathbb{E}\tilde{\boldsymbol{w}}_{i}\rightarrow 0$ as
$i\rightarrow\infty$) is therefore to select step-sizes $\{\mu_k\}$
that ensure \cite{Cattivelli10}:
\begin{equation}\label{eq18}\boxed{
    \rho(\mathcal{B}) < 1}
\end{equation}
where $\rho(\cdot)$ denotes the spectral radius of its matrix
argument. Note that the coefficient matrices $\{\mathcal{B}\}$ that
control the evolution of $\mathbb{E}\tilde{\boldsymbol{w}}_{i}$ are
different in the cases listed in Table II. These differences lead to
interesting conclusions.

\subsubsection*{\underline{B.1. Comparison of Mean Stability}}
To begin with, the matrix $\mathcal{B}$ is block diagonal in the
non-cooperative case and equal to
\begin{equation}\label{eq45}
    \mathcal{B}_\text{ncop}=I_{NM}-\mathcal{M}\mathcal{R}.
\end{equation}
Therefore, for each of the individual nodes to be stable in the
mean, it is necessary and sufficient that the step-sizes $\{\mu_k\}$
be selected to satisfy
\begin{equation}\label{eq138}
    \rho(\mathcal{B}_\text{ncop})=
    \max_{1\leq k\leq N}\rho(I_M-\mu_kR_{u,k})<1
\end{equation}
since the matrices $\mathcal{M}$ from (\ref{eq84}) and $\mathcal{R}$
from (\ref{eq120}) are block diagonal. Condition (\ref{eq138}) is
equivalent to
\begin{equation}\label{eq57}\text{(stability in the non-cooperative
case)}\quad\boxed{
    0 < \mu_k < \frac{2}{\lambda_{\max}(R_{u,k})}}
    \quad \text{for }k=1,2,\ldots,N
\end{equation}
where $\lambda_{\max}(\cdot)$ denotes the maximum eigenvalue of its
Hermitian matrix argument. Condition (\ref{eq57}) guarantees that
when each node acts individually and applies the LMS recursion
(\ref{eq56}), then the mean of its weight error vector will tend
asymptotically to zero. That is, by selecting the step-sizes to
satisfy (\ref{eq57}), \emph{all} individual nodes will be stable in
the mean.

Now consider the matrix $\mathcal{B}$ in the consensus case; it is
equal to
\begin{equation}\label{eq46}
    \mathcal{B}_\text{cons}=\mathcal{A}^T-\mathcal{M}\mathcal{R}.
\end{equation}
It is seen in this case that the stability of
$\mathcal{B}_\text{cons}$ depends on $\mathcal{A}$. The fact that
the stability of the consensus strategy is sensitive to the choice
of the combination matrix is known in the consensus literature for
the conventional implementation for computing averages and which
does not involve streaming data or gradient noise
\cite{Degroot74,Berger81}. Here, we are studying the more demanding
case of the single time-scale consensus iteration (\ref{eq7}) in the
presence of both \emph{noisy} and \emph{streaming} data. It is clear
from (\ref{eq46}) that the choice of $A$ can destroy the stability
of the consensus network even when the step-sizes are chosen
according to (\ref{eq57}) and all nodes are stable on their own.
This behavior does not occur for diffusion networks where the
matrices $\{\mathcal{B}\}$ for the ATC and CTA diffusion strategies
are instead given by
\begin{equation}\label{eq121}
    \mathcal{B}_\text{atc} =
    \mathcal{A}^T(I_{NM}-\mathcal{M}\mathcal{R})\text{ and }
    \mathcal{B}_\text{cta} =
    (I_{NM}-\mathcal{M}\mathcal{R})\mathcal{A}^T.
\end{equation}
The following result clarifies these statements.
\begin{thm}[Spectral properties of $\mathcal{B}$]\label{thm_1}
It holds that
\begin{equation}\label{eq51}
    \rho(\mathcal{B}_{\rm{atc}})=\rho(\mathcal{B}_{\rm{cta}})
     \leq \rho(\mathcal{B}_{\rm{ncop}})
\end{equation}
irrespective of the choice of the left-stochastic matrices $A$.
Moreover, if the combination matrix $A$ is symmetric, then the
eigenvalues of $\mathcal{B}_{\rm{cons}}$ are less than or equal to
the corresponding eigenvalues of $\mathcal{B}_{\rm{ncop}}$, i.e.,
\begin{equation}\label{eq33}
    \lambda_l(\mathcal{B}_{\rm{cons}})\leq \lambda_l(\mathcal{B}_{\rm{ncop}})
    \quad \text{for $l=1,2,\ldots,NM$}
\end{equation}
where the eigenvalues $\{\lambda_l(\cdot)\}$ are arranged in
decreasing order, i.e., $\lambda_{l_1}(\cdot)\geq
\lambda_{l_2}(\cdot)$ if $l_1\leq l_2$.
\end{thm}
\begin{proof}
See Appendix \ref{app_A}.
\end{proof}
Result (\ref{eq51}) establishes the important conclusion that the
coefficient matrix $\mathcal{B}$ for the diffusion strategies is
stable whenever $\mathcal{B}_\text{ncop}$ (or, from (\ref{eq138}),
each of the matrices $\{I_{M}-\mu_k R_{u,k}\}$) is stable;
\emph{this conclusion is independent of} $A$. The stability of the
matrices $\{I_{M}-\mu_k R_{u,k}\}$ is ensured by any step-size
satisfying (\ref{eq57}). Therefore, stability of the individual
nodes will always guarantee the stability of $\mathcal{B}$ in the
ATC and CTA diffusion cases, \emph{regardless} of the choice of $A$.
This is not the case for the consensus strategy (\ref{eq7}); even
when the step-sizes $\{\mu_k\}$ are selected to satisfy (\ref{eq57})
so that all individual nodes are mean stable, the matrix
$\mathcal{B}_{\rm cons}$ can still be unstable depending on the
choice of $A$ (and, therefore, on the network topology as well).
Therefore, if we start from a collection of nodes that are behaving
in a stable manner on their own, and if we connect them through a
topology and then apply consensus to solve the same estimation
problem through cooperation, then the network may end up being
unstable and the estimation task can fail drastically (see Fig.
\ref{Fig_3} further ahead). Moreover, it is further shown in
Appendix A that when $A$ is symmetric, the consensus strategy is
mean-stable for step-sizes satisfying:
\begin{equation}\label{eq17}
    0 < \mu_k <
    \frac{1+\lambda_{\min}(A)}{\lambda_{\max}(R_{u,k})}\quad
    \text{for }k=1,2,\ldots,N.
\end{equation}
Note from (\ref{eq50}) that since $A$ is a left-stochastic matrix,
its spectral radius is equal to one and one of its eigenvalues is
also equal to one \cite{Horn85}, i.e., $\lambda_1(A)=\rho(A)=1$.
This implies that the upper bound in (\ref{eq17}) is less than the
upper bound in (\ref{eq57}) so that diffusion networks are stable
over a wider range of step-sizes. Actually, the upper bound in
(\ref{eq17}) can be much smaller than the one in (\ref{eq57}) or
even zero because $\lambda_{\min}(A)$ can be negative or equal to
$-1$.

What if some of the nodes are unstable in the mean to begin with?
How would the behavior of the diffusion and consensus strategies
differ? Assume that there is at least one individual unstable node,
i.e., $\lambda_l(\mathcal{B}_\text{ncop})\leq -1$ for some $l$ so
that $\rho(\mathcal{B}_\text{ncop})\geq 1$. Then, we observe from
(\ref{eq51}) that the spectral radius of $\mathcal{B}_\text{atc}$
can still be smaller than one even if
$\rho(\mathcal{B}_\text{ncop})\geq 1$. It follows that even if some
individual node is unstable, the diffusion strategies can still be
stable if we properly choose $A$. In other words, diffusion
cooperation has a stabilizing effect on the network. In contrast, if
there is at least one individual unstable node and the combination
matrix $A$ is symmetric, then from (\ref{eq33}), no matter how we
choose $A$, the $\rho(\mathcal{B}_\text{cons})$ will be larger than
or equal to one and the consensus network will be unstable.

The above results suggest that fusing results from neighborhoods
according to the consensus strategy (\ref{eq7}) is not necessarily
the best thing to do because it can lead to instability and
catastrophic failure. On the other hand, fusing the results from
neighbors via diffusion ensures stability regardless of the
topology.

\subsubsection*{\underline{B.2. Example: Two-Node Networks}}
To illustrate these important observations, let us consider an
example consisting of two cooperating nodes; in this case, it is
possible to carry out the calculations analytically in order to
highlight the various patterns of behavior. Later, in the
simulations section, we illustrate the behavior for networks with
multiple nodes. Thus, consider a network consisting of $N=2$ nodes.
For simplicity, we assume the weight vector $w^\circ$ is a scalar,
and $R_{u,1}=\sigma^2_{u,1}$ and $R_{u,2}=\sigma^2_{u,2}$. Without
loss of generality, we assume $\mu_1\sigma^2_{u,1}\leq
\mu_2\sigma^2_{u,2}$. The combination matrix for this example is of
the form (Fig. \ref{Fig_3}):
\begin{equation}\label{eq13}
    A^T = \begin{bmatrix}
    1-a & a \\
    b & 1-b
    \end{bmatrix}.
\end{equation}
with $a,b\in[0,1]$. When desired, a symmetric $A$ can be selected by
simply setting $a=b$. Then, using (\ref{eq13}), we get
\begin{align}\label{eq16}
    \mathcal{B}_\text{atc}
    &=\begin{bmatrix}
    (1-\mu_1\sigma^2_{u,1})(1-a) & (1-\mu_2\sigma^2_{u,2})a\\
    (1-\mu_1\sigma^2_{u,1})b & (1-\mu_2\sigma^2_{u,2})(1-b)
    \end{bmatrix}\\
    \mathcal{B}_\text{cons}
    &=\begin{bmatrix}
    1-a-\mu_1\sigma^2_{u,1} & a\\
    b & 1-b-\mu_2\sigma^2_{u,2}
    \end{bmatrix}.
\end{align}
We first assume that
\begin{equation}\label{eq133}
    0<\mu_1\sigma^2_{u,1}\leq \mu_2\sigma^2_{u,2}<2
\end{equation}
so that both individual nodes are stable in the mean by virtue of
(\ref{eq57}). Then, by Theorem 1, the ATC diffusion network will
also be stable in the mean for any choice of the parameters
$\{a,b\}$. We now verify that there are choices for $\{a,b\}$ that
will turn the consensus network unstable. Specifically, we verify
below that if $a$ and $b$ happen to satisfy
\begin{equation}\label{eq44}
    a+b \geq 2-\mu_1\sigma^2_{u,1}
\end{equation}
then consensus will lead to unstable network behavior even though
both individual nodes are stable. Indeed, note first that the
minimum eigenvalue of $\mathcal{B}_\text{cons}$ is given by:
\begin{equation}
    \lambda_{\min}(\mathcal{B}_\text{cons})=
    \frac{(2-a-b-\mu_1\sigma^2_{u,1}-\mu_2\sigma^2_{u,2})-\sqrt{D}}{2}
\end{equation}
where
\begin{equation}\label{eq77}
\begin{aligned}
    D &\triangleq (-a+b-\mu_1\sigma^2_{u,1}+\mu_2\sigma^2_{u,2})^2+4ab\\
    &= (a+b+\mu_1\sigma^2_{u,1}-\mu_2\sigma^2_{u,2})^2+
    4b(\mu_2\sigma^2_{u,2}-\mu_1\sigma^2_{u,1}).
\end{aligned}
\end{equation}
From the first equality of (\ref{eq77}), we know that $D\geq 0$ and,
hence, $\lambda_{\min}(\mathcal{B}_\text{cons})$ is real. When
(\ref{eq133})-(\ref{eq44}) are satisfied, we have that
$(a+b+\mu_1\sigma^2_{u,1}-\mu_2\sigma^2_{u,2})$ and
$4b(\mu_2\sigma^2_{u,2}-\mu_1\sigma^2_{u,1})$ in the second equality
of (\ref{eq77}) are nonnegative. It follows that the consensus
network is unstable since
\begin{align}
    \lambda_{\min}(\mathcal{B}_\text{cons})\leq
    \frac{(2-a-b-\mu_1\sigma^2_{u,1}-\mu_2\sigma^2_{u,2})-
    (a+b+\mu_1\sigma^2_{u,1}-\mu_2\sigma^2_{u,2})}{2}
    \leq -1.\label{eq134}
\end{align}
In Fig. \ref{Fig_3}(a), we set $\mu_1\sigma^2_{u,1}=0.4$ and
$\mu_2\sigma^2_{u,2}=0.6$ so that each individual node is stable. If
we now set $a=b=0.85$, then (\ref{eq44}) is satisfied and the
consensus strategy becomes unstable.

\begin{figure}
\centering
\includegraphics[width=34em]{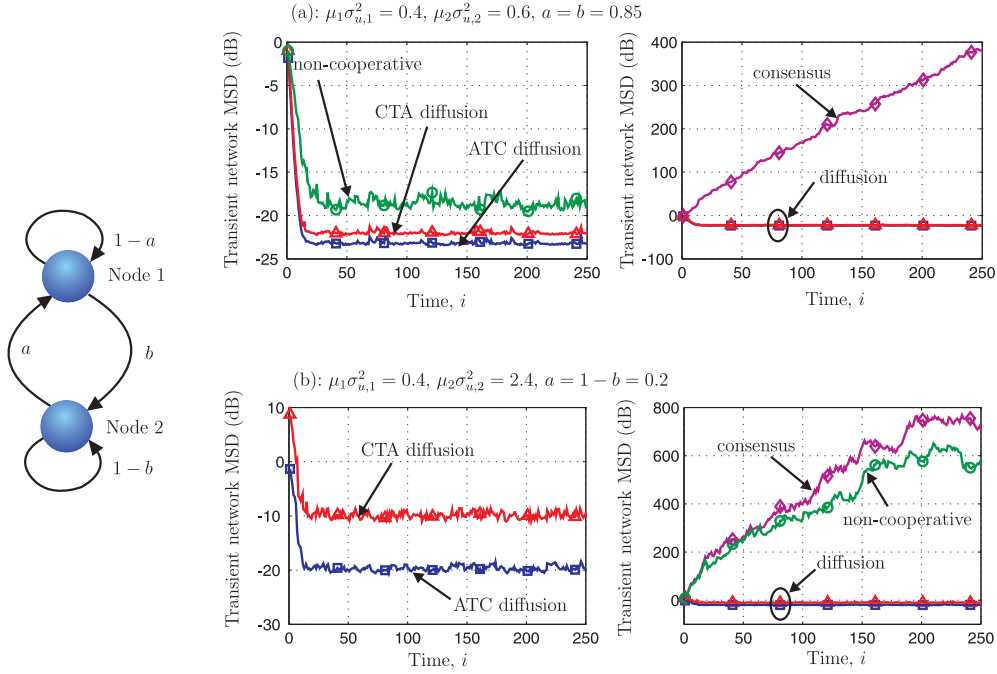}
\caption{Transient network MSD over time with $N=2$. (a)
$\mu_1\sigma^2_{u,1}=0.4$, $\mu_2\sigma^2_{u,2}=0.6$, and
$a=b=0.85$. As seen in the right plot, the consensus strategy is
unstable even when the individual nodes are stable. (b)
$\mu_1\sigma^2_{u,1}=0.4$, $\mu_2\sigma^2_{u,2}=2.4$, and
$a=1-b=0.2$ so that node $2$ is unstable. As seen in the right plot,
the diffusion strategies are able to stabilize the network even when
the non-cooperative and consensus strategies are unstable.}
\label{Fig_3}
\end{figure}

Next, we consider an example satisfying
\begin{equation}\label{eq15}
    0< \mu_1\sigma^2_{u,1}< 2 \leq \mu_2\sigma^2_{u,2}
\end{equation}
so that node $1$ is still stable, whereas node $2$ becomes unstable.
From the first equality of (\ref{eq77}), we again conclude that
\begin{equation}
\begin{aligned}
    \lambda_{\min}(\mathcal{B}_\text{cons})&\leq
    \frac{(2-a-b-\mu_1\sigma^2_{u,1}-\mu_2\sigma^2_{u,2})-
    |-a+b-\mu_1\sigma^2_{u,1}+\mu_2\sigma^2_{u,2}|}{2}\\
    &=\begin{cases}
    1-a-\mu_1\sigma^2_{u,1},
    &\text{if $a+\mu_1\sigma^2_{u,1}\geq b+\mu_2\sigma^2_{u,2}$}\\
    1-b-\mu_2\sigma^2_{u,2},
    &\text{otherwise}
    \end{cases}\\
    &\leq -1.
\end{aligned}
\end{equation}
That is, in this second case, no matter how we choose the parameters
$\{a,b\}$, the consensus network is always unstable. In contrast,
the diffusion network is able to stabilize the network. To see this,
we set $b=1-a$ so that the eigenvalues of $\mathcal{B}_\text{atc}$
in (\ref{eq16}) are
$\{0,1-\mu_1\sigma^2_{u,1}-(\mu_2\sigma^2_{u,2}-\mu_1\sigma^2_{u,1})a\}$.
Some algebra shows that the diffusion network is stable if $a$
satisfies
\begin{equation}\label{eq91}
    0 \leq a <
    \frac{2-\mu_1\sigma^2_{u,1}}{\mu_2\sigma^2_{u,2}-\mu_1\sigma^2_{u,1}}.
\end{equation}
In Fig. \ref{Fig_3}(b), we set $\mu_1\sigma^2_{u,1}=0.4$ and
$\mu_1\sigma^2_{u,1}=2.4$ so that node $1$ is stable, but node $2$
is unstable. If we now set $a=1-b=0.2$, then (\ref{eq91}) is
satisfied and the diffusion strategies become stable even when the
non-cooperative and consensus strategies are unstable.

\subsection{Mean-Square Stability}
We now examine the stability in the mean-square sense of the
consensus and diffusion strategies. Let $\Sigma$ denote an arbitrary
nonnegative-definite matrix that we are free to choose. From
(\ref{eq12}), we get the following weighted variance relation for
sufficiently small step-sizes:
\begin{equation} \label{eq28}\boxed{
    \mathbb{E}\|\tilde{\boldsymbol{w}}_{i}\|^2_{\Sigma}\approx
    \mathbb{E}\|\tilde{\boldsymbol{w}}_{i-1}\|^2_{\mathcal{B}^*\Sigma\mathcal{B}}+
    \text{Tr}(\Sigma \mathcal{Y})}
\end{equation}
where the notation $\|x\|^2_\Sigma$ denotes the weighted square
quantity $x^*\Sigma x$ and $\mathcal{Y}\triangleq
\mathbb{E}\boldsymbol{y}_i\boldsymbol{y}_i^*$ appears in Table II
with the covariance matrix $\mathcal{S}$ defined by:
\begin{align}
    \mathcal{S}\triangleq \mathbb{E}\boldsymbol{s}_i\boldsymbol{s}_i^*
    =\text{diag}\{\sigma^2_{v,1}R_{u,1},\sigma^2_{v,2}R_{u,2},
    \ldots,\sigma^2_{v,N}R_{u,N}\}.
\end{align}
As shown in \cite{Cattivelli10,Zhao12,Sayed13}, step-sizes that
satisfy (\ref{eq18}) and are sufficiently small will also ensure
mean-square stability of the network (namely,
$\mathbb{E}\|\tilde{\boldsymbol{w}}_{i}\|^2_{\Sigma}\rightarrow
c<\infty$ as $i\rightarrow \infty$). Therefore, we find again that,
for infinitesimally small step-sizes, the mean-square stability of
consensus networks is sensitive to the choice of $A$, whereas the
mean-square stability of diffusion networks is not affected by $A$.
In the next section, we will examine $\rho(\mathcal{B})$ more
closely for the various strategies listed in Table II and establish
that diffusion networks are not only more stable than consensus
networks but also lead to better mean-square-error performance as
well.

\subsection{Mean-Square Deviation}
The mean-square deviation (MSD) measure is used to assess how well
the nodes in the network estimate the weight vector, $w^\circ$. The
MSD at node $k$ is defined as follows:
\begin{equation}\label{eq65}
    \text{MSD}_k\triangleq\lim_{i\rightarrow\infty}
    \mathbb{E}\|\tilde{\boldsymbol{w}}_{k,i}\|^2
\end{equation}
where $\|\cdot\|$ denotes the Euclidean norm for vectors. The
network MSD is defined as the average MSD across the network, i.e.,
\begin{equation}\label{eq85}
    \text{MSD}\triangleq
    \frac{1}{N}\sum_{k=1}^{N}\text{MSD}_k.
\end{equation}
Iterating (\ref{eq28}), we can obtain a series expression for the
network MSD as:
\begin{equation} \label{eq31}\boxed{
\begin{aligned}
    \text{MSD}=\frac{1}{N}\sum_{j=0}^\infty\text{Tr}
    [\mathcal{B}^j\mathcal{Y}(\mathcal{B}^*)^j]
\end{aligned}}.
\end{equation}
We can also obtain a series expansion for the MSD at each individual
node $k$ as follows:
\begin{equation}\label{eq67}\boxed{
\begin{aligned}
    \text{MSD}_k = \sum_{j=0}^\infty\text{Tr}
    \left[(e_k^T\otimes I_{M})\cdot\mathcal{B}^j\mathcal{Y}
    \mathcal{B}^{*j}\cdot (e_k\otimes I_{M})\right]
\end{aligned}}
\end{equation}
where $e_k$ denotes the $k$th column of the identity matrix $I_N$.
Expressions (\ref{eq67})-(\ref{eq31}) relate the MSDs directly to
the quantities $\{\mathcal{B},\mathcal{Y}\}$ from Table II.

\section{Comparison of Mean-Square Performance for Homogeneous Agents}
In the previous section, we compared the stability of the various
estimation strategies in the mean and mean-square senses. In
particular, we established that stability of the individual nodes
ensures stability of diffusion networks irrespective of the
combination topology. In the sequel, we shall assume that the
step-sizes are sufficiently small so that conditions (\ref{eq57})
and (\ref{eq17}) hold and the diffusion and consensus networks are
stable in the mean and mean-square sense; as well as the individual
nodes. Under these conditions, the networks achieve steady-state
operation. We now use the MSD expressions derived above to establish
that ATC diffusion achieves lower (and, hence, better) MSD values
than the consensus, CTA, and non-cooperative strategies. In this
way, diffusion strategies do not only ensure stability of the
cooperative behavior but they also lead to improved
mean-square-error performance. We establish these results under the
following reasonable condition.
\begin{assumption}
All nodes in the network use the same step-size, $\mu_k=\mu$, and
they observe data arising from the same covariance data so that
$R_{u,k}=R_u$ for all $k$. In other words, we are dealing with a
network of homogeneous nodes interacting with each other. In this
way, it is possible to quantify the differences in performance
without biasing the results by differences in the adaptation
mechanism (step-sizes)  or in the covariance matrices of the
regression data at the nodes.
\end{assumption}
\noindent Under Assumption 1, it holds that $\mathcal{M} = \mu
I_{NM}$ and $\mathcal{R} = I_N\otimes R_u$, and thus the matrices
$\mathcal{B}$ and $\mathcal{Y}$ in Table II reduce to the
expressions shown in Table III, where we introduced the diagonal
matrix
\begin{equation}\label{eq137}
    \Sigma_v
    \triangleq\text{diag}\{\sigma^2_{v,1},\sigma^2_{v,2},\ldots,\sigma^2_{v,N}\}>0.
\end{equation}
Note that the ATC and CTA diffusion strategies now have the same
coefficient matrix $\mathcal{B}$. We explain in the sequel the terms
that appear in the last row of Table III.

\begin{table}
\centering \caption{Variables for cooperative and non-cooperative
implementations when $\mu_k=\mu$ and $R_{u,k}=R_u$.}
\renewcommand{\arraystretch}{1.2}
\begin{tabular}{|c|c|c|c|c|}
\hline & \textbf{ATC diffusion} (\ref{eq49})& \textbf{CTA diffusion}
(\ref{eq52})& \textbf{Consensus} (\ref{eq7})
& \textbf{Non-cooperative} (\ref{eq56})\\
\hline\hline $\mathcal{B}$ &
    $A^T\otimes I_{M}-A^T\otimes \mu R_u$ &
    $A^T\otimes I_{M}-A^T\otimes \mu R_u$ &
    $A^T\otimes I_{M}-I_N\otimes \mu R_u$ &
    $I_N\otimes I_{M}-I_N\otimes \mu R_u$\\
\hline $\lambda_{l,m}(\mathcal{B})$ &
    $\lambda_l(A)(1-\mu\lambda_m(R_u))$ &
    $\lambda_l(A)(1-\mu\lambda_m(R_u))$ &
    $\lambda_l(A)-\mu\lambda_m(R_u)$ &
    $1-\mu\lambda_m(R_u)$\\
\hline $\mathcal{Y}$ &
    $\mu^2(A^T\Sigma_v A)\otimes R_u$ &
    $\mu^2\Sigma_v\otimes R_u$ &
    $\mu^2\Sigma_v\otimes R_u$ &
    $\mu^2\Sigma_v\otimes R_u$\\
\hline $s^{b*}_{l,m}\mathcal{Y}s^b_{l,m}$ &
    $\mu^2\lambda_m(R_u)|\lambda_l(A)|^2\cdot s^*_l\Sigma_vs_l$ &
    $\mu^2\lambda_m(R_u)\cdot s^*_l\Sigma_vs_l$ &
    $\mu^2\lambda_m(R_u)\cdot s^*_l\Sigma_vs_l$ &
    $\mu^2\lambda_m(R_u)\cdot s^*_l\Sigma_vs_l$ \\
\hline
\end{tabular}
\end{table}

\subsection{Spectral Properties of $\mathcal{B}$}
As mentioned before, the stability and mean-square-error performance
of the various algorithms depend on the corresponding matrix
$\mathcal{B}$; therefore, in this section, we examine more closely
the eigen-structure of $\mathcal{B}$. For the distributed strategies
(diffusion and consensus), the eigen-structure of $\mathcal{B}$ will
depend on the combination matrix $A$. Thus, let $r_l$ and $s_l$
($l=1,2,\ldots,N$) denote an arbitrary pair of right and left
eigenvectors of $A^T$ corresponding to the eigenvalue
$\lambda_l(A)$. That is,
\begin{equation}\label{eq135}
    A^T r_l = \lambda_l(A)r_l \text{ and }
    s_l^* A^T = \lambda_l(A) s_l^*.
\end{equation}
We scale the vectors $r_l$ and $s_l$ to satisfy:
\begin{equation}\label{eq23}
    \|r_l\|=1\text{ and }s^*_lr_l=1 \text{ for all $l$.}
\end{equation}
Recall that $\lambda_1(A)=\rho(A)=1$. Furthermore, we let $z_m$
($m=1,2,\ldots,M$) denote the eigenvector of the covariance matrix
$R_u$ that is associated with the eigenvalue $\lambda_m(R_u)$. That
is,
\begin{equation}
    R_uz_m = \lambda_m(R_u)z_m.
\end{equation}
Since $R_u$ is Hermitian and positive-definite, the $\{z_m\}$ are
orthonormal, i.e., $z^*_{m_2}z_{m_1}=\delta_{m_1m_2}$, and the
$\{\lambda_m(R_u)\}$ are positive. The following result describes
the eigen-structure of the matrix $\mathcal{B}$ in terms of the
eigen-structures of $\{A^T,R_u\}$ for the diffusion and consensus
algorithms of Table III. Note that the results for any of these
distributed strategies collapse to the result for the
non-cooperative strategy when we set $\lambda_l(A)=1$ for all $l$.
\begin{lem}[Eigen-structure of $\mathcal{B}$ under diffusion and
consensus]\label{lem_1} The matrices $\{\mathcal{B}\}$ appearing in
Table III for the diffusion and consensus strategies have right and
left eigenvectors $\{r^b_{l,m},s^{b}_{l,m}\}$ given by:
\begin{align}\label{eq6}
    r^b_{l,m} = r_l \otimes z_m \text{ and }
    s^{b}_{l,m} = s_l \otimes z_m
\end{align}
with the corresponding eigenvalues, $\lambda_{l,m}(\mathcal{B})$,
shown in Table III for $l=1,2,\ldots,N$ and $m=1,2,\ldots,M$. Note
that while the eigenvectors are the same for the diffusion and
consensuses strategies, the corresponding eigenvalues are different.
\end{lem}
\begin{proof}
We only consider the diffusion case and denote its coefficient
matrix by $\mathcal{B}_\text{diff}=A^T\otimes I_{M}-A^T\otimes \mu
R_u$; the same argument applies to the consensus strategy. We
multiply $\mathcal{B}_\text{diff}$ by the $r^b_{l,m}$ defined in
(\ref{eq6}) from the right and obtain
\begin{equation}
\begin{aligned}
    \mathcal{B}_\text{diff}\cdot r^b_{l,m}&=(A^T\otimes I_{M}-A^T\otimes \mu R_u)\cdot
    (r_l \otimes z_m)\\
    &=\lambda_l(A)\cdot(r_l \otimes z_m)-
    \lambda_l(A)\cdot\mu\lambda_m(R_u)\cdot(r_l \otimes z_m)\\
    &=\lambda_l(A)(1-\mu\lambda_m(R_u))\cdot r^b_{l,m}
\end{aligned}
\end{equation}
where we used the Kronecker product property $(A\otimes B)(C\otimes
D) = AC\otimes BD$ for matrices $\{A,B,C,D\}$ of compatible
dimensions \cite{Sayed08}. In a similar manner, we can verify that
$\mathcal{B}_\text{diff}$ has left eigenvector $s^b_{l,m}$ defined
in (\ref{eq6}) with the corresponding eigenvalue
$\lambda_{l,m}(\mathcal{B})$ from Table III.
\end{proof}
\begin{thm}[Spectral radius of $\mathcal{B}$ under diffusion and
consensus]\label{thm_2} Under Assumption 1, it holds that
\begin{equation}\label{eq36}
    \rho(\mathcal{B}_{\rm{diff}})=\rho(\mathcal{B}_{\rm{ncop}})
    \leq \rho(\mathcal{B}_{\rm{cons}})
\end{equation}
where equality holds if $A=I_N$ or when the step-size satisfies:
\begin{equation}\label{eq19}
    0 < \mu\leq \min_{l\neq 1}\frac{1-|\lambda_l(A)|}
    {\lambda_{\min}(R_u)+\lambda_{\max}(R_u)}.
\end{equation}
\end{thm}
\begin{proof}
See Appendix \ref{app_B}.
\end{proof}
Note that the upper bound in (\ref{eq19}) is even smaller than the
one in (\ref{eq17}) and, therefore, can again be very small or even
zero. It follows that there is generally a wide range of step-sizes
over which $\rho(\mathcal{B}_{\rm{cons}})$ is greater than
$\rho(\mathcal{B}_{\rm{diff}})$. When this happens, the convergence
rate of diffusion networks is superior to the convergence rate of
consensus networks; in particular, the quantities
$\mathbb{E}\tilde{\boldsymbol{w}}_i$ and
$\mathbb{E}\|\tilde{\boldsymbol{w}}_i\|^2$ will converge faster
towards their steady-state values over diffusion networks than over
consensus networks.

\subsection{Network MSD Performance}
We now compare the MSD performance. Note that the expressions for
the individual MSD in (\ref{eq67}) and the network MSD in
(\ref{eq31}) depend on $\mathcal{B}$ in a nontrivial manner. To
simplify these MSD expressions, we introduce the following
assumption on the combination matrix.
\begin{assumption}
The combination matrix $A$ is diagonalizable, i.e., there exists an
invertible matrix $U$ and a diagonal matrix $\Lambda$ such that
\begin{equation}\label{eq136}
    A^T=U\Lambda U^{-1}
\end{equation}
with
\begin{align}
    U &= \begin{bmatrix}r_1 & r_2 & \cdots & r_N \end{bmatrix},\quad
    U^{-1} ={\rm col}\{s^*_1,s^*_2, \ldots,  s^*_N\}\\
    \Lambda &= {\rm diag}\{\lambda_1(A),\lambda_2(A), \ldots
    ,\lambda_N(A)\}.
\end{align}
That is, the columns of $U$ consist of the right eigenvectors of
$A^T$ and the rows of $U^{-1}$ consist of the left eigenvectors of
$A^T$, as defined by (\ref{eq135}).
\end{assumption}
\noindent Note that, besides condition (\ref{eq23}), it follows from
Assumption 2 that $s^*_{l_2}r_{l_1}=\delta_{l_1l_2}$. Furthermore,
any \emph{symmetric} combination matrix $A$ is diagonalizable and
therefore satisfies condition (\ref{eq136}) automatically. Actually,
when $A$ is symmetric, more can be said about its eigenvectors. In
that case, the matrix $U$ will be orthogonal so that $U^{-1}=U^T$
and it will further hold that $r_{l_2}^*r_{l_1}=\delta_{l_1l_2}$.
Assumption 2 allows the analysis to apply to important cases in
which $A$ is not necessarily symmetric but is still diagonalizable
(such as when $A$ is constructed according to the uniform rule by
assigning to the links of node $k$ weights that are equal to the
inverse of its degree, $n_k$). We can now simplify the MSD
expressions by using the eigen-decomposition of $\mathcal{B}$ from
Lemma 1 and the above eigen-decomposition of $A$.
\begin{lem}[MSD expressions] \label{lem_2}
The MSD at node $k$ from (\ref{eq67}) can be expressed as:
\begin{align}\label{eq73}
    {\rm{MSD}}_k=\sum_{l_1=1}^N\sum_{l_2=1}^N\sum_{m=1}^M
    \frac{(e_k^Tr_{l_1})\cdot s^{b*}_{l_1,m}\mathcal{Y}s^b_{l_2,m}\cdot (r^*_{l_2}e_k)}
    {1-\lambda_{l_1,m}(\mathcal{B})\lambda^*_{l_2,m}(\mathcal{B})}.
\end{align}
Furthermore, if the right eigenvectors $\{r_l\}$ of $A^T$ are
approximately orthonormal, i.e.,
\begin{equation}\label{eq14}
    r^*_{l_2}r_{l_1} \approx \delta_{l_1l_2}
\end{equation}
then the network MSD from (\ref{eq31}) can be approximated by:
\begin{align} \label{eq39}
    {\rm{MSD}}\approx\sum_{l=1}^N\sum_{m=1}^M
    \frac{s^{b*}_{l,m}\mathcal{Y}s^b_{l,m}}
    {N\cdot\left(1-|\lambda_{l,m}(\mathcal{B})|^2\right)}.
\end{align}
\end{lem}
\begin{proof}
See Appendix \ref{app_C}.
\end{proof}
\noindent Note that any symmetric combination matrix $A$ satisfies
condition (\ref{eq14}) since, as mentioned above, its right
eigenvectors can be chosen to be orthonormal.

Using the expressions for $\lambda_{l,m}(\mathcal{B})$ and
$s^{b*}_{l,m}\mathcal{Y}s^b_{l,m}$ from Table III and substituting
into (\ref{eq39}), we can obtain the network MSD expressions for the
various strategies. The following result shows how these MSD values
compare to each other.
\begin{thm}[Comparing network MSDs] \label{thm_3}
If condition (\ref{eq14}) is satisfied, then the ATC diffusion
strategy achieves the lowest network MSD in comparison to the other
strategies (CTA diffusion, consensus, and non-cooperative). More
specifically, it holds that
\begin{align}\label{eq92}
    {\rm{MSD}}_{\rm{atc}} &\leq {\rm{MSD}}_{\rm{cta}}
    \leq {\rm{MSD}_{\rm{ncop}}}\\
    {\rm{MSD}}_{\rm{atc}} &\leq {\rm{MSD}}_{\rm{cons}}.\label{eq93}
\end{align}
Furthermore, if $1\leq\mu \lambda_{\min}(R_u)<2$, the consensus
strategy is the worst even in comparison to the non-cooperative
strategy:
\begin{equation}
    {\rm{MSD}}_{\rm{atc}} \leq {\rm{MSD}}_{\rm{cta}}
    \leq {\rm{MSD}}_{\rm{ncop}}\leq {\rm{MSD}}_{\rm{cons}}.
\end{equation}
\end{thm}
\begin{proof}
See Appendix \ref{app_D}.
\end{proof}
Therefore, the ATC diffusion strategy outperforms consensus, CTA
diffusion, and non-cooperative strategies when condition
(\ref{eq14}) is satisfied. However, the relation among
MSD$_\text{cta}$, MSD$_\text{cons}$, and MSD$_\text{ncop}$ depends
on the combination matrix $A$. To illustrate this fact, we
reconsider the two-node network from Section III.B with
$\sigma^2_{u,1}=\sigma^2_{u,2}=\sigma^2_u$, $\mu_1=\mu_2=\mu$, and
$0<\mu\sigma^2_u<1$. Furthermore, to ensure the stability of the
consensus strategy and from (\ref{eq44}), the parameters $\{a,b\}$
in (\ref{eq13}) are now chosen to satisfy $a+b<2-\mu\sigma^2_u$. In
this case, the eigenvalues of the combination matrix $A$ in
(\ref{eq13}) are $\{1,1-a-b\}$. It can be verified from (\ref{eq39})
and Table III that the CTA diffusion strategy achieves lower network
MSD (better mean-square performance) than the consensus strategy if
\begin{equation}\label{eq88}
\begin{cases}
    \text{MSD}_\text{cons} \leq \text{MSD}_\text{cta},
    &\text{if $0\leq a+b \leq
    \frac{2(1-\mu\sigma^2_u)}{2-\mu\sigma^2_u}$}\\
    \text{MSD}_\text{cons} \geq \text{MSD}_\text{cta},
    &\text{if $\frac{2(1-\mu\sigma^2_u)}{2-\mu\sigma^2_u}
    \leq a+b < 2-\mu\sigma^2_u$}
\end{cases}
\end{equation}
Similarly, the network MSDs of the consensus and non-cooperative
strategies have the following relation:
\begin{equation}\label{eq89}
\begin{cases}
    \text{MSD}_\text{cons} \leq \text{MSD}_\text{ncop},
    &\text{if $0\leq a+b \leq 2(1-\mu\sigma^2_u)$}\\
    \text{MSD}_\text{cons} \geq \text{MSD}_\text{ncop},
    &\text{if $2(1-\mu\sigma^2_u)\leq a+b < 2-\mu\sigma^2_u$}
\end{cases}
\end{equation}
Combining (\ref{eq88})-(\ref{eq89}), we can divide the $a\times b$
plane into three regions, as shown in Fig. \ref{Fig_4}, where each
region corresponds to one possible relation among MSD$_\text{cta}$,
MSD$_\text{cons}$, and MSD$_\text{ncop}$.

\begin{figure}
\centering
\includegraphics[width=38em]{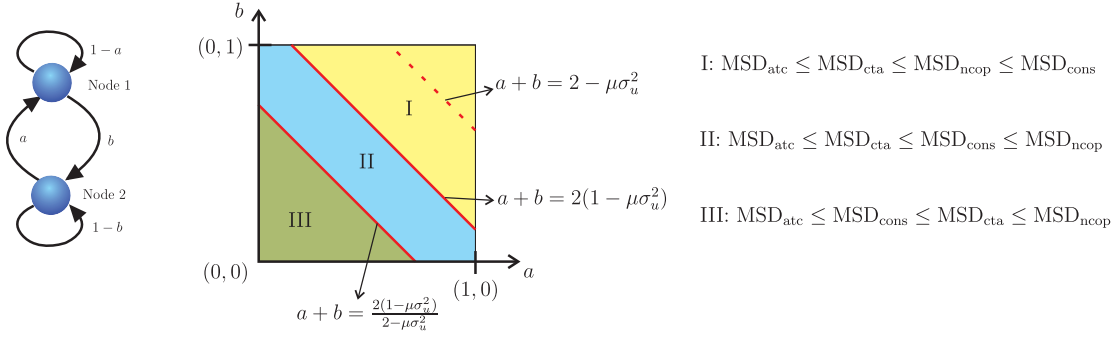}
\caption{Network MSD comparison with $N=2$ and $\mu\sigma^2_u=0.4$.
The consensus strategy is unstable when the parameters $a$ and $b$
lie above the dashed line in region I.} \label{Fig_4}
\end{figure}

\subsection{MSD of Individual Nodes}
In Theorem \ref{thm_3}, we established that the ATC diffusion
strategy performs the best in terms of the average network MSD. It
is still not clear how well the individual nodes perform under each
strategy. It is generally more challenging to compare diffusion and
consensus strategies in terms of the MSDs of their individual nodes
due to the structure of the matrix $\mathcal{B}$ for the consensus
strategy. Nevertheless, this can be accomplished as follows. We
observe from (\ref{eq73}) and Table III that the
$\{\text{MSD}_{k}\}$ for the CTA diffusion and consensus strategies
differ only in the value of $\lambda_{l,m}(\mathcal{B})$. From Table
III, the difference between the values of
$\lambda_{l,m}(\mathcal{B})$ for these two strategies is
\begin{equation}\label{eq47}
\begin{aligned}
    \lambda_{l,m}(\mathcal{B}_\text{cta})-\lambda_{l,m}(\mathcal{B}_\text{cons})
    =\mu\lambda_m(R_u)\cdot(1-\lambda_l(A))
    =\mathcal{O}(\mu)
\end{aligned}
\end{equation}
where the term $\mathcal{O}(\mu)$ denotes a factor that is of the
order of the step-size $\mu$. It follows that for sufficiently small
step-sizes, expression (\ref{eq47}) is close to zero and the CTA
diffusion and consensus strategies will exhibit similar MSDs at the
individual nodes, i.e.,
$\text{MSD}_{\text{cta},k}\approx\text{MSD}_{\text{cons},k}$ for all
$k$. As a result, in the following, we only compare
$\text{MSD}_{\text{atc},k}$, $\text{MSD}_{\text{cta},k}$, and
$\text{MSD}_{\text{ncop},k}$. In particular, we will show that under
certain conditions on the combination matrix $A$, the ATC diffusion
strategy continues to perform the best in terms of the MSD at the
individual nodes in comparison to the other strategies. To do so,
starting from (\ref{eq73}) and the expressions for
$\{\lambda_{l,k}(\mathcal{B}),\mathcal{Y}\}$ in Table III, we can
express the MSD at node $k$ for the ATC diffusion strategy as:
\begin{equation}
\begin{aligned}
    {\rm{MSD}}_{\text{atc},k}=\sum_{m=1}^M\mu^2\lambda_m(R_u)\sum_{l_1,l_2=1}^N
    \frac{\lambda_{l_1}(A)\lambda^*_{l_2}(A)\cdot
    (e_k^Tr_{l_1}s^{*}_{l_1}\Sigma_vs_{l_2}r^*_{l_2}e_k)}
    {1-\lambda_{l_1}(A)\lambda^*_{l_2}(A)\cdot(1-\mu\lambda_m(R_u))^2}
    \triangleq \sum_{m=1}^M{\rm{MSD}}_{\text{atc},k}(m)
\end{aligned}
\end{equation}
where we introduced the notation MSD$_{\text{atc},k}(m)$ to denote
the MSD component at node $k$ that is contributed by the $m$th
eigenvalue of $R_u$, i.e.,
\begin{equation}\label{eq90}
    {\rm{MSD}}_{\text{atc},k}(m)=\mu^2\lambda_m(R_u)\sum_{l_1,l_2=1}^N
    \frac{\lambda_{l_1}(A)\lambda^*_{l_2}(A)\cdot
    (e_k^Tr_{l_1}s^{*}_{l_1}\Sigma_vs_{l_2}r^*_{l_2}e_k)}
    {1-\lambda_{l_1}(A)\lambda^*_{l_2}(A)\cdot(1-\mu\lambda_m(R_u))^2}.
\end{equation}
In a similar vein, we can define the corresponding MSD$_{k}(m)$
terms for the other strategies. We list these terms in Table IV in
two equivalent forms (we will use the series form later). We first
have the following useful preliminary result.

\begin{table}
\centering \caption{Expressions for MSD$_k(m)$ in series form and
eigen-form.}
\renewcommand{\arraystretch}{2}
\begin{tabular}{|c||c|c|}
\hline \multirow{2}{2cm}{\textbf{ATC Diffusion} (\ref{eq49})} &
Series form & {\normalsize
    $\mu^2\lambda_m(R_u)\sum_{j=0}^\infty(1-\mu\lambda_m(R_u))^{2j}\cdot
    e_k^TA^{T(j+1)}\Sigma_vA^{j+1}e_k$}\\ \cline{2-3}
    & Eigen-form & {\normalsize
    $\mu^2\lambda_m(R_u)\sum_{l_1,l_2=1}^N
    \frac{\lambda_{l_1}(A)\lambda^*_{l_2}(A)\cdot
    (e_k^Tr_{l_1}s^{*}_{l_1}\Sigma_vs_{l_2}r^*_{l_2}e_k)}
    {1-\lambda_{l_1}(A)\lambda^*_{l_2}(A)\cdot(1-\mu\lambda_m(R_u))^2}$ }\\
\hline\hline \multirow{2}{2cm}{\textbf{CTA Diffusion} (\ref{eq52})}
& Series form & {\normalsize
    $\mu^2\lambda_m(R_u)\sum_{j=0}^\infty(1-\mu\lambda_m(R_u))^{2j}\cdot
    e_k^TA^{Tj}\Sigma_vA^{j}e_k$}\\ \cline{2-3}
    & Eigen-form & {\normalsize
    $\mu^2\lambda_m(R_u)\sum_{l_1,l_2=1}^N
    \frac{e_k^Tr_{l_1}s^{*}_{l_1}\Sigma_vs_{l_2}r^*_{l_2}e_k}
    {1-\lambda_{l_1}(A)\lambda^*_{l_2}(A)\cdot(1-\mu\lambda_m(R_u))^2}$} \\
\hline\hline \multirow{2}{2cm}{\textbf{Non-cooperative}
(\ref{eq56})} & Series form & {\normalsize
    $\mu^2\lambda_m(R_u)\sum_{j=0}^\infty(1-\mu\lambda_m(R_u))^{2j}\cdot
    e_k^T \Sigma_v e_k$ }\\ \cline{2-3}
    & Eigen-form & {\normalsize
    $\mu^2\lambda_m(R_u)\sum_{l_1,l_2=1}^N
    \frac{e^T_kr_{l_1}s^{*}_{l_1}\Sigma_vs_{l_2}r^*_{l_2}e_k}
    {1-(1-\mu\lambda_m(R_u))^2}$} \\
\hline
\end{tabular}
\end{table}

\begin{lem}[Useful comparisons] \label{lem_3}
The following ratios are positive and independent of the node index
$k$:
\begin{align}\label{eq97}
    \frac{{\rm{MSD}}_{{\rm{ncop}},k}(m)-{\rm{MSD}}_{{\rm{atc}},k}(m)}
    {{\rm{MSD}}_{{\rm{ncop}},k}(m)-{\rm{MSD}}_{{\rm{cta}},k}(m)}
    &=\frac{1}{(1-\mu\lambda_m(R_u))^2}>0\\
    \frac{{\rm{MSD}}_{{\rm{ncop}},k}(m)-{\rm{MSD}}_{{\rm{atc}},k}(m)}
    {{\rm{MSD}}_{{\rm{cta}},k}(m)-{\rm{MSD}}_{{\rm{atc}},k}(m)}
    &=\frac{1}{1-(1-\mu\lambda_m(R_u))^2}>0. \label{eq98}
\end{align}
\end{lem}
\begin{proof}
From the eigen-forms of $\{\text{MSD}_{k}(m)\}$ in Table IV, the
differences between MSD$_{\text{atc},k}(m)$,
MSD$_{\text{cta},k}(m)$, and MSD$_{\text{ncop},k}(m)$ are given by:
\begin{align}\label{eq94}
    {\rm{MSD}}_{\text{ncop},k}(m)-{\rm{MSD}}_{\text{atc},k}(m)&=
    \frac{\mu^2\lambda_m(R_u)}{1-(1-\mu\lambda_m(R_u))^2}\cdot c_k(m)\\
    {\rm{MSD}}_{\text{ncop},k}(m)-{\rm{MSD}}_{\text{cta},k}(m)&=
    \frac{\mu^2\lambda_m(R_u)\cdot(1-\mu\lambda_m(R_u))^2}
    {1-(1-\mu\lambda_m(R_u))^2}\cdot c_k(m) \label{eq95}\\
    {\rm{MSD}}_{\text{cta},k}(m)-{\rm{MSD}}_{\text{atc},k}(m)&=
    \mu^2\lambda_m(R_u)\cdot c_k(m)\label{eq96}
\end{align}
where
\begin{equation}
    c_k(m) =\sum_{l_1,l_2=1}^N
    \frac{\left[1-\lambda_{l_1}(A)\lambda^*_{l_2}(A)\right]\cdot
    (e^T_kr_{l_1}s^{*}_{l_1}\Sigma_vs_{l_2}r^*_{l_2}e_k)}
    {1-\lambda_{l_1}(A)\lambda^*_{l_2}(A)\cdot(1-\mu\lambda_m(R_u))^2}.
\end{equation}
Then, dividing (\ref{eq94}) by (\ref{eq95}) and (\ref{eq94}) by
(\ref{eq96}), we arrive at (\ref{eq97})-(\ref{eq98}).
\end{proof}
\begin{lem}[Useful ordering] \label{lem_4}
The relation among ${\rm{MSD}}_{{\rm{atc}},k}(m)$,
${\rm{MSD}}_{{\rm{cta}},k}(m)$, and ${\rm{MSD}}_{{\rm{ncop}},k}(m)$
is either
\begin{equation}\label{eq9}
    {\rm{MSD}}_{{\rm{atc}},k}(m)\leq {\rm{MSD}}_{{\rm{cta}},k}(m)
    \leq {\rm{MSD}}_{{\rm{ncop}},k}(m)
\end{equation}
or
\begin{equation}\label{eq99}
    {\rm{MSD}}_{{\rm{atc}},k}(m)\geq {\rm{MSD}}_{{\rm{cta}},k}(m)
    \geq {\rm{MSD}}_{{\rm{ncop}},k}(m).
\end{equation}
\end{lem}
\begin{proof}
Assume first that ${\rm{MSD}}_{{\rm{atc}},k}(m) \leq
{\rm{MSD}}_{{\rm{ncop}},k}(m)$. Then, using (\ref{eq97}), we get
${\rm{MSD}}_{{\rm{ncop}},k}(m)-{\rm{MSD}}_{{\rm{cta}},k}(m) \geq 0$.
Similarly, from (\ref{eq98}), we get
${\rm{MSD}}_{{\rm{cta}},k}(m)-{\rm{MSD}}_{{\rm{atc}},k}(m)\geq 0$.
We conclude that relation (\ref{eq9}) holds in this case. Assume
instead that ${\rm{MSD}}_{{\rm{atc}},k}(m) \geq
{\rm{MSD}}_{{\rm{ncop}},k}(m)$. Then, a similar argument will show
that (\ref{eq99}) should hold.
\end{proof}
The above result is useful since it allows us to deduce the relation
among ${\rm{MSD}}_{{\rm{atc}},k}(m)$,
${\rm{MSD}}_{{\rm{cta}},k}(m)$, and ${\rm{MSD}}_{{\rm{ncop}},k}(m)$
by only knowing the relation between any two of them. To proceed, we
note that we can alternatively express the MSD$_{k}(m)$ terms in an
equivalent series form. For example, expression (\ref{eq90}) can be
written as:
\begin{align}
    {\rm{MSD}}_{\text{atc},k}(m)&=\mu^2\lambda_m(R_u)\sum_{j=0}^{\infty}\sum_{l_1,l_2=1}^N
    (1-\mu\lambda_m(R_u))^{2j}\cdot\lambda^{j+1}_{l_1}(A)\cdot\lambda^{*(j+1)}_{l_2}(A)\cdot
    (e_k^Tr_{l_1}s^{*}_{l_1}\Sigma_vs_{l_2}r^*_{l_2}e_k)\notag \\
    &=\mu^2\lambda_m(R_u)\sum_{j=0}^{\infty}
    (1-\mu\lambda_m(R_u))^{2j}\cdot
    e_k^T\left(\sum_{l_1=1}^N\lambda^{j+1}_{l_1}(A)r_{l_1}s^{*}_{l_1}\right)
    \Sigma_v\left(\sum_{l_2=1}^N\lambda^{*(j+1)}_{l_2}(A)s_{l_2}r^*_{l_2}\right)e_k \notag\\
    &=\mu^2\lambda_m(R_u)\sum_{j=0}^\infty(1-\mu\lambda_m(R_u))^{2j}\cdot
    e_k^TA^{T(j+1)}\Sigma_vA^{j+1}e_k.
\end{align}
In a similar manner, we can obtain the corresponding MSD$_{k}(m)$
series forms for the other strategies and we list these in Table IV.
In the following, we provide conditions to guarantee that the
individual node performance in the ATC diffusion strategy
outperforms the other strategies.
\begin{thm}[Comparing individual MSDs] \label{thm_4}
If the combination matrix $A$ satisfies
\begin{equation}\label{eq101}
    \Sigma_v-A^T\Sigma_vA \geq 0
\end{equation}
where $\Sigma_v$ is the noise variance (diagonal) matrix defined by
(\ref{eq137}), then:
\begin{equation}\label{eq100}
    {\rm{MSD}}_{{\rm{atc}},k} \leq {\rm{MSD}}_{{\rm{cta}},k}
    \leq {\rm{MSD}}_{{\rm{ncop}},k}.
\end{equation}
\end{thm}
\begin{proof}
From the series forms of $\{\text{MSD}_k(m)\}$ in Table IV, the
difference
${\rm{MSD}}_{\text{cta},k}(m)-{\rm{MSD}}_{\text{atc},k}(m)$ is given
by:
\begin{equation}
\begin{aligned}
    {\rm{MSD}}_{\text{cta},k}(m)-{\rm{MSD}}_{\text{atc},k}(m)=
    \mu^2\lambda_m(R_u)\sum_{j=0}^\infty(1-\mu\lambda_m(R_u))^{2j}e^T_kA^{Tj}
    \left(\Sigma_v-A^T\Sigma_vA\right)A^{j}e_k.
\end{aligned}
\end{equation}
Since $\Sigma_v-A^T\Sigma_vA \geq 0$, we conclude that
${\rm{MSD}}_{\text{cta},k}(m)\geq{\rm{MSD}}_{\text{atc},k}(m)$ for
all $m$. Then, applying Lemma \ref{lem_4}, we obtain relation
(\ref{eq100}).
\end{proof}
Condition (\ref{eq101}) essentially means that the combination
matrix $A$ should not magnify the noise effect across the network.
However, in general, condition (\ref{eq101}) is restrictive in the
sense that over the set of feasible diagonalizable left-stochastic
matrices $A$ satisfying $a_{l,k}=0$ if $l\notin\mathcal{N}_k$, the
set of combination matrices $A$ satisfying (\ref{eq101}) can be
small. We illustrate this situation by reconsidering the two-node
network (\ref{eq13}) for which
\begin{equation}\label{eq102}
    \Sigma_v-A^T\Sigma_vA=\begin{bmatrix}
    2at-a^2(1+t) & -(1-a)bt-a(1-b) \\
    -(1-a)bt-a(1-b) & 2b-b^2(1+t)
    \end{bmatrix}
\end{equation}
where $t=\sigma^2_{v,1}/\sigma^2_{v,2}$ denotes the ratio of noise
variances at nodes 1 and 2. Note from
\begin{equation}\label{eq123}
    \text{det}(\Sigma_v-A^T\Sigma_vA)=-(a-bt)^2\leq 0
\end{equation}
that equality holds in (\ref{eq123}) if, and only if,
\begin{equation}\label{eq81}
    a=tb.
\end{equation}
That is, when $a\neq tb$, the matrix $(\Sigma_v-A^T\Sigma_vA)$ has
two eigenvalues with different signs. Thus, the only way to ensure
$\Sigma_v-A^T\Sigma_vA\geq 0$ in this case is to set $a=tb$ and,
thus, the matrix $(\Sigma_v-A^T\Sigma_vA)$ will have at least one
eigenvalue at zero since its determinant will be zero. To ensure
$\Sigma_v-A^T\Sigma_vA\geq 0$, its other eigenvalue, which is equal
to $b(1+t^2)(2-b-bt)$, needs to be greater than or equal to zero. It
follows that $b$ must satisfy:
\begin{equation}\label{eq74}
    0\leq b\leq \frac{2}{1+t}.
\end{equation}
Moreover, since $a$ and $b$ must lie within the interval $[0,1]$, we
conclude from (\ref{eq81}) that $b$ must also satisfy:
\begin{equation}\label{eq75}
    0\leq b \leq \min\{1,1/t\}.
\end{equation}
It can be verified that condition (\ref{eq75}) implies condition
(\ref{eq74}) since $\min\{1,1/t\}\leq 2/(1+t)$. That is, for any
left-stochastic matrix $A$ from (\ref{eq13}) satisfying $a=tb$ and
(\ref{eq75}), relation (\ref{eq100}) holds and both nodes improve
their own MSDs by employing the diffusion strategies. Note that
condition (\ref{eq81}) represents a line segment in the unit square
$a,b\in[0,1]$ (see Fig. \ref{Fig_5}). In the following, we relax
condition (\ref{eq101}) with a mild constraint on the network
topology.

In addition to Assumption 2, we further assume that the combination
matrix $A$ is primitive (also called regular). This means that there
exists an integer $j$ such that the $j$th power of $A$ has positive
entries, $[A^j]_{l,k}>0$ for all $l$ and $k$ \cite{Horn85}. We
remark that for any connected network (where a path always exists
between any two arbitrary nodes), if the combination weights
$\{a_{l,k}\}$ satisfy $a_{l,k}>0$ for $l\in\mathcal{N}_k$, then $A$
is primitive. Now, since $A$ is primitive, it follows from the
Perron-Frobenius Theorem \cite{Horn85} that $(A^T)^j$ converges to
the rank-one matrix:
\begin{equation}\label{eq80}
    \lim_{j\rightarrow \infty} (A^T)^j=r_1s^T_1.
\end{equation}
From (\ref{eq50}) and (\ref{eq23}), $r_1$ and $s_1$ satisfy:
\begin{equation}\label{eq68}
    r_1 = \frac{\mathds{1}}{\sqrt{N}} \text{ and }
    \frac{s^T_1\mathds{1}}{\sqrt{N}}=1.
\end{equation}
\begin{thm}[Comparing individual MSDs for regular networks] \label{thm_5}
For any primitive and diagonalizable combination matrix $A$, if
\begin{equation}\label{eq103}
    \frac{s^T_1\Sigma_v s_1}{N} < \sigma^2_{v,k}
\end{equation}
for all $k$, then there exists $\mu^\circ>0$ so that for any
step-size $\mu$ satisfying $0 < \mu \leq \mu^\circ$, it holds:
\begin{align}\label{eq104}
    {\rm{MSD}}_{{\rm{atc}},k} < {\rm{MSD}}_{{\rm{cta}},k}
    < {\rm{MSD}}_{{\rm{ncop}},k}.
\end{align}
\end{thm}
\begin{proof}
See Appendix \ref{app_E}.
\end{proof}
\begin{figure}
\centering
\includegraphics[width=36em]{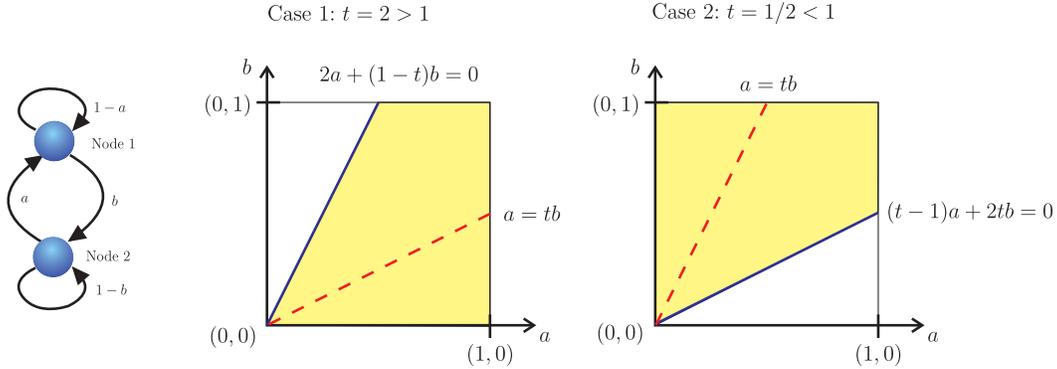}
\caption{Comparison of individual node MSD using $N=2$ and
$t=\sigma^2_{v,1}/\sigma^2_{v,2}$. There exists a step-size region
such that ${\rm{MSD}}_{{\rm{atc}},k} < {\rm{MSD}}_{{\rm{cta}},k}<
{\rm{MSD}}_{{\rm{ncop}},k}$ for $k=1,2$ when the parameters $a$ and
$b$ lie in the shaded regions. The dashed lines indicate condition
(\ref{eq81}).} \label{Fig_5}
\end{figure}
We show in Appendix \ref{app_F} that for any primitive $A$,
condition (\ref{eq101}) implies condition (\ref{eq103}). To
illustrate these two conditions, we consider again the two-node
network. It can be verified that $s^T_1$ for $A^T$ in (\ref{eq13})
has the form $s^T_1=\begin{bmatrix}\sqrt{2}b/(a+b) & \sqrt{2}a/(a+b)
\end{bmatrix}$. Then, some algebra shows that condition
(\ref{eq103}) becomes
\begin{equation}\label{eq106}
    (t-1)a+2bt > 0\text{ and }
    2a+(1-t)b > 0.
\end{equation}
Recall that $t=\sigma^2_{v,1}/\sigma^2_{v,2}$. We illustrate
condition (\ref{eq106}), along with condition (\ref{eq81}), in Fig.
\ref{Fig_5}. We observe that condition (\ref{eq81}), shown as the
dashed lines, is contained in condition (\ref{eq106}), shown as the
shaded regions, and that compared to condition (\ref{eq81}),
condition (\ref{eq106}) enlarges the region of $A$ for which the ATC
diffusion strategy performs the best in terms of the individual MSD
performance.

\section{Simulation Results}
We consider a network with $20$ nodes and random topology. The
regression covariance matrix $R_u$ is diagonal with entries randomly
generated from $[2,4]$, and the noise variances $\{\sigma^2_{v,k}\}$
are randomly generated over $[-30,-10]$ dB (see Fig. \ref{Fig_6}).
The network estimates a $10\times1$ (i.e., $M=10$) unknown vector
$w^\circ$ with every entry equal to $1/\sqrt{10}$.

\begin{figure}
\centering
\includegraphics[width=28em]{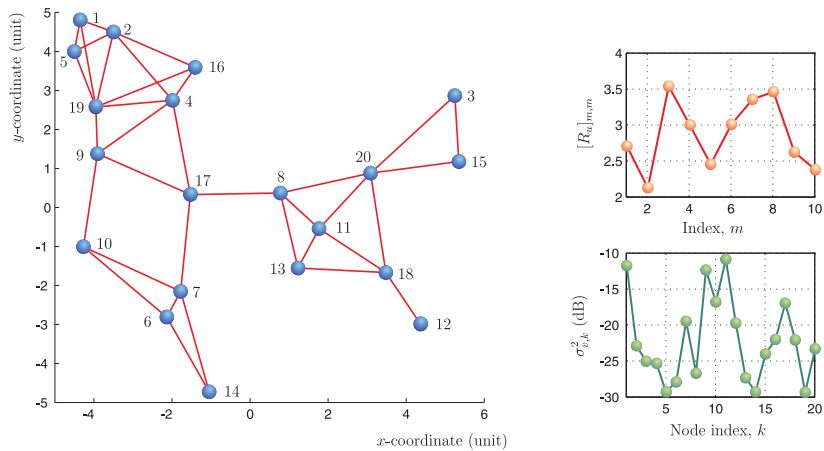}
\caption{Network topology and noise and data power profiles at the
nodes. The number next to a node denotes the node index.}
\label{Fig_6}
\end{figure}

\begin{table}
\centering \caption{Combination rules used in the simulations,
$a_{l,k}=0 \text{ if } l\notin{\mathcal{N}}_k$}
\begin{tabular}{|l|l|}
\hline  \textbf{Name} & \textbf{Rule} \\
\hline\hline Relative-variance \cite{Zhao12} & $a_{l,k}=
    \sigma^{-2}_{v,l}/\sum_{j\in\mathcal{N}_k}\sigma^{-2}_{v,j}$\\
\hline Uniform \cite{Sayed13} & $a_{l,k}=1/n_k$\\
\hline Metropolis \cite{Xiao05} &
    $a_{l,k}=
    \begin{cases}
    1-\sum_{j\neq k}a_{k,j}, &\text{if $l=k$}\\
    1/\max\{n_k,n_l\}, &\text{if $l\in\mathcal{N}_k\setminus\{k\}$}
    \end{cases}$\\
\hline
\end{tabular}
\end{table}

\begin{figure}
\centering
\includegraphics[width=30em]{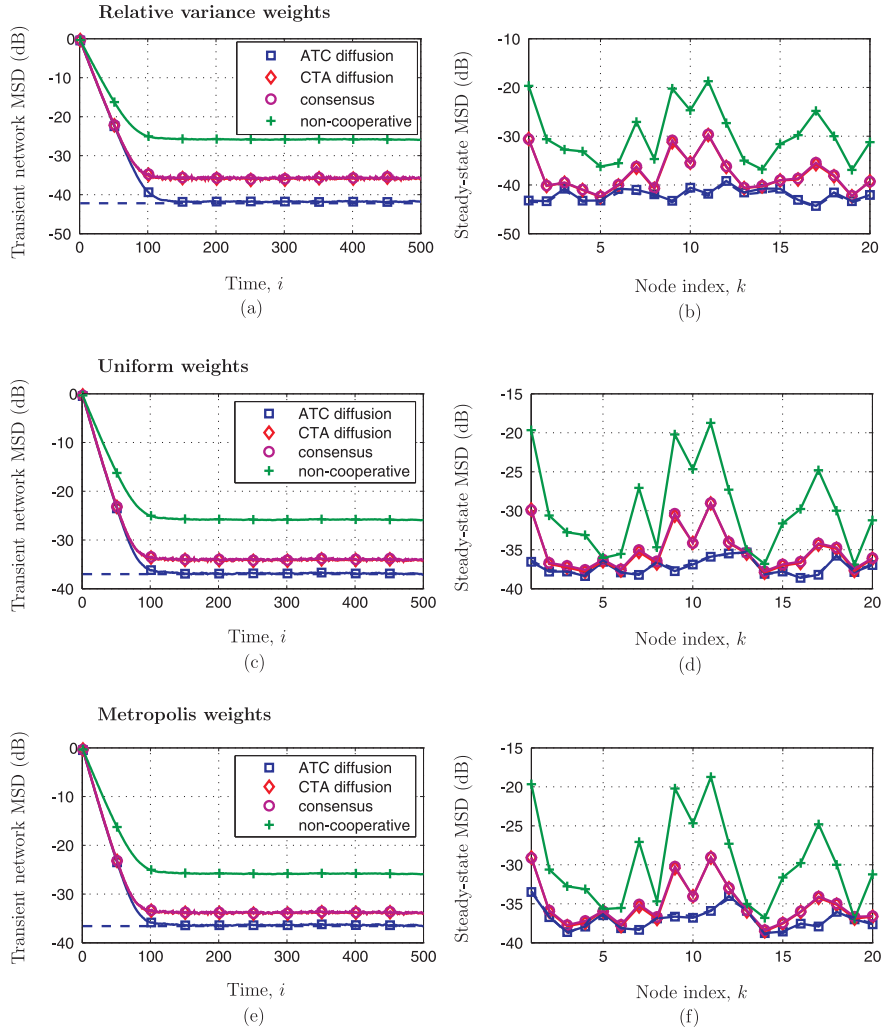}
\caption{Transient network MSD over time (left, with peak values
normalized to 0dB) and steady-state MSD at the individual nodes
(right) for (a)-(b) the relative-variance, (c)-(d) uniform, and
(e)-(f) Metropolis rules. The dashed lines on the left/right hand
side indicate the theoretical network/individual MSD from
(\ref{eq39})/(\ref{eq73}) for the ATC diffusion strategy.}
\label{Fig_1}
\end{figure}

The transient network MSD over time is shown on the left hand side
of Fig. \ref{Fig_1} with three possible combination rules:
relative-variance \cite{Zhao12}, uniform \cite{Sayed13}, and
Metropolis \cite{Xiao05} (see Table V). Note that the matrix $A$ for
the Metropolis rule is symmetric. The step-size $\mu$ is set to
$\mu=0.02$. We observe that, as expected, the ATC diffusion strategy
outperforms the other strategies, especially for the
relative-variance rule. It also suggests that some conventional
choices of combination weights, such as the Metropolis rule, may not
be the most suitable for adaptation in the presence of both noisy
and streaming data because such weights do not take into account the
noise profile across the nodes (see, e.g., \cite{Zhao12,Sayed13} for
more details on this issue). We further show the steady-state MSD at
the individual nodes on the right hand side of Fig. \ref{Fig_1}. We
observe that the ATC diffusion strategy achieves the lowest MSD at
each node in comparison to the other strategies. These observations
are in agreement with the results predicted by the theoretical
analysis. The theoretical expressions for MSDs from
(\ref{eq67})-(\ref{eq31}) are also depicted in Fig. \ref{Fig_1} for
the ATC diffusion strategy and match well with simulations.

\begin{figure}
\centering
\includegraphics[width=30em]{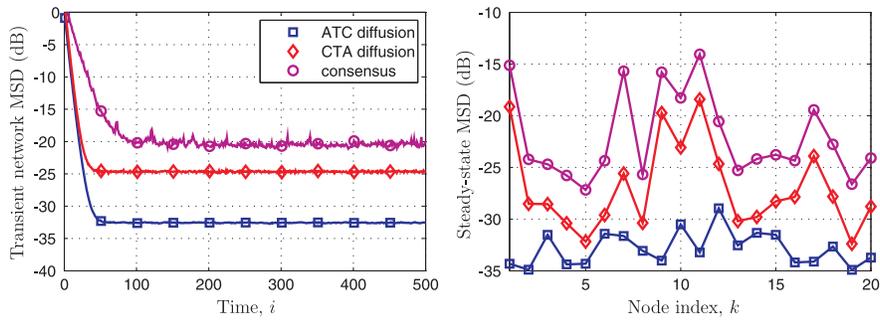}
\caption{Transient network MSD over time (left) and steady-state MSD
at the individual nodes (right) for the relative-variance
combination rule using $\mu = 0.075$.} \label{Fig_7}
\end{figure}

We further compare the mean-square performance of the distributed
strategies for larger step-sizes. We set the step-size to
$\mu=0.075$ and use the relative-variance combination rule. The
transient network MSD over time is shown on the left hand side of
Fig. \ref{Fig_7}. We observe that the ATC and CTA diffusion
strategies have the same convergence rate and converge faster than
the consensus strategy. Moreover, the diffusion strategies achieve
lower network MSD than the consensus strategy. We also show the
steady-state MSD at the individual nodes on the right hand side of
Fig. \ref{Fig_7}. We see again that ATC diffusion performs the best
in comparison to the other strategies at each individual node.

\section{Concluding Remarks}
We compared analytically several cooperative estimation strategies,
including ATC diffusion, CTA diffusion, and consensus for
distributed estimation over networks. The results show that
diffusion networks are more stable than consensus networks.
Moreover, the stability of diffusion networks is independent of the
combination weights, whereas consensus networks can become unstable
even if all individual nodes are stable. Furthermore, in
steady-state, the ATC diffusion strategy performs the best not only
in terms of the network MSD, but also in terms of the MSDs at the
individual nodes.

\appendices
\section{Proof of Theorem \ref{thm_1}}\label{app_A}
First, note that the matrices $\{\mathcal{B}\}$ for the ATC and CTA
diffusion strategies given by (\ref{eq121}) have the same
eigenvalues (and, therefore, the same spectral radius) because for
any matrices $X$ and $Y$ of compatible dimensions, the matrix
products $XY$ and $YX$ have the same eigenvalues \cite{Horn85}. So
let us evaluate the spectral radius of $\mathcal{B}_\text{atc}$. To
do so, we introduce a convenient block matrix norm, and denote it by
$\|\cdot\|_b$; it is defined as follows. Let $\mathcal{X}$ be an
$N\times N$ block matrix with blocks of size $M\times M$ each. Its
block matrix norm is defined as:
\begin{equation}\label{eq78}
    \|\mathcal{X}\|_{b}\triangleq\max_{1\leq k \leq N}
    \left(\sum_{l=1}^N \|\mathcal{X}_{k,l} \|_2\right)
\end{equation}
where $\mathcal{X}_{k,l}$ denotes the $(k,l)$th block of
$\mathcal{X}$ and $\|\cdot\|_{2}$ denotes the 2-induced norm
(largest singular value) of its matrix argument. Now, since
$\{I_{NM},\mathcal{M},\mathcal{R}\}$ are block diagonal matrices,
the following property holds:
\begin{equation}
\begin{aligned}
    \|I_{NM}-\mathcal{M}\mathcal{R}\|_{b}=
    \max_{1\leq k \leq N}\|I_M-\mu_kR_{u,k}\|_2
    =\max_{1\leq k \leq N}\rho(I_M-\mu_kR_{u,k})
    =\rho(\mathcal{B}_\text{ncop})
\end{aligned}
\end{equation}
where we used the fact that the 2-induced norm of any Hermitian
matrix coincides with its spectral radius. In addition, since $A$ is
a left-stochastic matrix, it holds that
\begin{equation}
\begin{aligned}
    \|\mathcal{A}^T\|_{b}&=\max_{1\leq k \leq N}
    \left(\sum_{l=1}^N \|a_{l,k}I_M \|_2\right)
    =\max_{1\leq k \leq N}
    \left(\sum_{l=1}^N a_{l,k}\right)
    =1.
\end{aligned}
\end{equation}
Accordingly, using the fact that the spectral radius of a matrix is
upper bounded by any norm of the matrix \cite{Horn85}, we get:
\begin{align}
    \rho(\mathcal{B}_\text{atc})\leq
    \|\mathcal{A}^T(I_{NM}-\mathcal{M}
    \mathcal{R})\|_b
    \leq \|\mathcal{A}^T\|_b\cdot \|I_{NM}-\mathcal{M}
    \mathcal{R}\|_b
    =\rho(\mathcal{B}_\text{ncop})
\end{align}
which establishes (\ref{eq51}).

Now, assume $A$ is symmetric. Since it is also left-stochastic, it
follows that its eigenvalues are real and lie inside the interval
$[-1,1]$. Therefore, $(I_{NM}-\mathcal{A}^T)$ is
nonnegative-definite. Moreover, since $\mathcal{M}$ and
$\mathcal{R}$ commute, i.e.,
$\mathcal{R}\mathcal{M}=\mathcal{M}\mathcal{R}$, it can be verified
that $\mathcal{B}_\text{cons}$ in (\ref{eq46}) and
$\mathcal{B}_\text{ncop}$ in (\ref{eq45}) are Hermitian. In
addition, the matrices $\mathcal{B}_\text{cons}$ and
$\mathcal{B}_\text{ncop}$ are related as follows:
\begin{equation}
    \mathcal{B}_\text{ncop} = \mathcal{B}_\text{cons}+(I_{NM}-\mathcal{A}^T)
\end{equation}
with $(I_{NM}-\mathcal{A}^T)\geq 0$. Using Weyl's
Theorem\footnote{Let $\{D', D, \Delta D\}$ be $M\times M$ Hermitian
matrices with ordered eigenvalues $\{\lambda_m(D'),
\lambda_m(D),\lambda_m(\Delta D)\}$, i.e.,
$\lambda_1(D)\geq\lambda_2(D)\geq \ldots\geq\lambda_M(D)$, and
likewise for the eigenvalues of $\{D',\Delta D\}$. Weyl's Theorem
states that if $D'=D+\Delta D$, then
\begin{equation*}
    \lambda_m(D)+\lambda_M(\Delta D)\leq
    \lambda_m(D')\leq\lambda_m(D)+\lambda_1(\Delta D)
\end{equation*}
for $1\leq m \leq M$. When $\Delta D\geq 0$, it holds that
$\lambda_m(D')\geq\lambda_m(D)$.} \cite{Horn85}, we arrive at
(\ref{eq33}). Following a similar argument, it holds for symmetric
$A$ that
\begin{equation}
    \lambda_l\left\{\lambda_{\min}(A)\cdot
    I_{NM}-\mathcal{M}\mathcal{R}\right\}
    \leq \lambda_l(\mathcal{B}_\text{cons})\quad
    \text{for }l=1,2,\ldots,NM.
\end{equation}
Thus, the matrix $\mathcal{B}_\text{cons}$ is stable (namely,
$-1<\lambda_l(\mathcal{B}_\text{cons})<1$ for $l=1,2,\ldots,NM$) if
\begin{align}
    \lambda_l\left(\lambda_{\min}(A)\cdot
    I_{NM}-\mathcal{M}\mathcal{R}\right) &> -1\\
    \lambda_l(\mathcal{B}_\text{ncop}) &< 1
\end{align}
for $l=1,2,\ldots,NM$, or, equivalently,
\begin{align}
    \lambda_{\min}(A)-\mu_k\lambda_m(R_{u,k}) &> -1\\
    1-\mu_k\lambda_m(R_{u,k}) &< 1
\end{align}
for $k=1,2,\ldots,N$ and $m=1,2,\ldots,M$. We then arrive at
(\ref{eq17}).

\section{Proof of Theorem \ref{thm_2}} \label{app_B}
For the diffusion strategies, from Table III and since $\rho(A)=1$,
we have
\begin{equation}
    \rho(\mathcal{B}_\text{diff})=\rho[A^T\otimes(I_M-\mu R_u)]
    =\rho(A)\cdot \rho(I_M-\mu R_u) = \rho(I_M-\mu R_u) =
    \rho(\mathcal{B}_\text{ncop}).
\end{equation}
Moreover, since $1\in\{\lambda_l(A)\}$, we have
\begin{equation}\label{eq22}
    \rho(\mathcal{B}_\text{ncop})=\max_{1\leq m\leq M}|1-\mu \lambda_{m}(R_u)|
    \leq \max_{1\leq l\leq N}\max_{1\leq m\leq M}|\lambda_l(A)-\mu \lambda_{m}(R_u)|
    = \rho(\mathcal{B}_\text{cons})
\end{equation}
and we arrive at (\ref{eq36}). It is obvious that when $A=I_N$, then
equality in (\ref{eq22}) holds and $\rho(\mathcal{B}_\text{ncop})=
\rho(\mathcal{B}_\text{cons})$. We now consider the case when $A\neq
I_N$. Note that the spectral radius of $\mathcal{B}_\text{ncop}$ is
given by
\begin{equation}
    \rho(\mathcal{B}_\text{ncop}) = \max\{1-\mu\lambda_{\min}(R_u),
    -1+\mu\lambda_{\max}(R_u)\}.
\end{equation}
We first verify that equality in (\ref{eq22}) holds only when
$\rho(\mathcal{B}_\text{ncop})=1-\mu\lambda_{\min}(R_u)$. Indeed, if
$\rho(\mathcal{B}_\text{ncop})=-1+\mu\lambda_{\max}(R_u)\geq0$, we
have that $\mu\lambda_{\max}(R_u)\geq 1$ and we get from
(\ref{eq22}) that
\begin{align}
    \rho(\mathcal{B}_\text{cons})&=
    \max_{1\leq l\leq N}\max_{1\leq m\leq M}|\lambda_l(A)-\mu \lambda_{m}(R_u)|\notag\\
    &\geq
    |\lambda_l(A)-\mu \lambda_{\max}(R_u)|\notag\\
    &\geq|\mathrm{Re}\{\lambda_l(A)\}-\mu
    \lambda_{\max}(R_u)|\notag\\
    &=-\mathrm{Re}\{\lambda_l(A)\}+\mu \lambda_{\max}(R_u)
\end{align}
since $\mathrm{Re}\{\lambda_l(A)\}\leq 1$ where
$\mathrm{Re}\{\cdot\}$ denotes the real part of its argument. Since
$A\neq I_N$, there exists some $l$ such that
$\mathrm{Re}\{\lambda_l(A)\}<1$ and then
$\rho(\mathcal{B}_\text{cons})>-1+\mu\lambda_{\max}(R_u)=\rho(\mathcal{B}_\text{ncop})$.
Now, assume that
$\rho(\mathcal{B}_\text{ncop})=1-\mu\lambda_{\min}(R_u)$. Then,
equality in (\ref{eq22}) holds if
\begin{equation}\label{eq25}
    |\lambda_l(A)-\mu \lambda_{m}(R_u)|\leq \rho(\mathcal{B}_\text{ncop})
\end{equation}
for all $l$ and $m$. It is obvious that relation (\ref{eq25}) holds
for $l=1$ since $\lambda_1(A)=1$ and
\begin{align}
    \rho(\mathcal{B}_\text{ncop})&=
    \max_{1\leq m\leq M}|1-\mu \lambda_{m}(R_u)| \notag\\
    &\geq |\lambda_1(A)-\mu \lambda_{m}(R_u)|.
\end{align}
For $l=2,3,\ldots,N$, by the triangular inequality of norms, we have
that $|\lambda_l(A)-\mu \lambda_{m}(R_u)|\leq
|\lambda_l(A)|+\mu\lambda_{\max}(R_u)$. Hence, the inequality in
(\ref{eq25}) holds if
\begin{equation}
    |\lambda_l(A)|+\mu\lambda_{\max}(R_u)\leq 1-\mu\lambda_{\min}(R_u)
\end{equation}
for $l=2,3,\ldots,N$ and we arrive at (\ref{eq19}).

\section{Proof of Lemma \ref{lem_2}} \label{app_C}
From Lemma \ref{lem_1}, the eigen-decomposition for the matrix power
$\mathcal{B}^j$ is given by:
\begin{equation}\label{eq118}
    \mathcal{B}^j=\sum_{l=1}^N\sum_{m=1}^M\lambda^j_{l,m}(\mathcal{B})
    \cdot r^b_{l,m}s^{b*}_{l,m}.
\end{equation}
Using (\ref{eq118}), we can rewrite the MSD at node $k$ from
(\ref{eq67}) as:
\begin{equation}\label{eq62}
\begin{aligned}
    \text{MSD}_k&=
    \sum_{j=0}^\infty\sum_{l_1,l_2=1}^N\sum_{m_1,m_2=1}^M
    \text{Tr}\left[\lambda^j_{l_1,m_1}(\mathcal{B})
    \lambda^{*j}_{l_2,m_2}(\mathcal{B})\cdot
    (e_k^T\otimes I_{M}) \cdot r^b_{l_1,m_1}s^{b*}_{l_1,m_1}
    \mathcal{Y}s^{b}_{l_2,m_2}r^{b*}_{l_2,m_2}\cdot (e_k\otimes I_{M})\right]\\
    &=\sum_{l_1,l_2=1}^N\sum_{m_1,m_2=1}^M
    \frac{\left(r^{b*}_{l_2,m_2}(e_k\otimes I_{M})
    (e^T_k\otimes I_{M})r^b_{l_1,m_1}\right)\cdot
    \left(s^{b*}_{l_1,m_1} \mathcal{Y}s^{b}_{l_2,m_2}\right)}
    {1-\lambda_{l_1,m_1}(\mathcal{B})\lambda^*_{l_2,m_2}(\mathcal{B})}
\end{aligned}
\end{equation}
where we used $\text{Tr}(AB)=\text{Tr}(BA)$ and the expression for
the infinite sum of a geometric series. Using (\ref{eq6}), we have:
\begin{align}
    r^{b*}_{l_2,m_2}(e_k\otimes I_{M})
    (e^T_k\otimes I_{M})r^b_{l_1,m_1}
    =(r^*_{l_2}e_ke_k^Tr_{l_1})\otimes (z^*_{m_2}z_{m_1})
    =(r^*_{l_2}e_ke_k^Tr_{l_1})\cdot\delta_{m_1m_2} \label{eq119}
\end{align}
since the eigenvectors $\{z_m\}$ are orthonormal. Substituting
(\ref{eq119}) into (\ref{eq62}), we arrive at (\ref{eq73}).
Likewise, from (\ref{eq85}) and (\ref{eq73}), the network MSD is
given by
\begin{equation}
\begin{aligned}
    \text{MSD}=\frac{1}{N}\sum_{l_1,l_2=1}^N\sum_{m=1}^M
    \frac{\left(\sum_{k=1}^Nr^*_{l_2}e_ke_k^Tr_{l_1}\right)
    \cdot s^{b*}_{l_1,m}\mathcal{Y}s^b_{l_2,m}}
    {1-\lambda_{l_1,m}(\mathcal{B})\lambda^*_{l_2,m}(\mathcal{B})}.
\end{aligned}
\end{equation}
From assumption (\ref{eq14}), we can establish (\ref{eq39}) since
\begin{equation}\label{eq117}
    \sum_{k=1}^Nr^*_{l_2}e_ke_k^Tr_{l_1}=r^*_{l_2}\cdot I_N\cdot r_{l_1}\approx
    \delta_{l_1l_2}.
\end{equation}

\section{Proof of Theorem \ref{thm_3}} \label{app_D}
We first verify that ${\rm{MSD}}_{\rm{atc}} \leq
{\rm{MSD}}_{\rm{cta}}$, ${\rm{MSD}}_{\rm{cta}} \leq
{\rm{MSD}}_{\rm{ncop}}$, and ${\rm{MSD}}_{\rm{atc}} \leq
{\rm{MSD}}_{\rm{cons}}$. We show the result by verifying that the
individual terms on the right hand side of (\ref{eq39}) for the
various strategies have the same ordering. That is, from
(\ref{eq39}) and Table III, we verify that the following ratios,
which correspond to ${\rm{MSD}}_{\rm{atc}} \leq
{\rm{MSD}}_{\rm{cta}}$, ${\rm{MSD}}_{\rm{cta}} \leq
{\rm{MSD}}_{\rm{ncop}}$, and ${\rm{MSD}}_{\rm{atc}} \leq
{\rm{MSD}}_{\rm{cons}}$, respectively, are upper bounded by one:
\begin{align}\label{eq64}
    |\lambda_l(A)|^2&\leq 1\\
    \frac{1-(1-\mu\lambda_m(R_u))^2}
    {1-|\lambda_l(A)|^2\cdot(1-\mu\lambda_m(R_u))^2}&\leq 1\label{eq24} \\
    \frac{|\lambda_l(A)|^2\left(1-|\lambda_l(A)-\mu\lambda_m(R_u)|^2\right)}
    {1-|\lambda_l(A)|^2\cdot (1-\mu\lambda_m(R_u))^2}
    &\leq 1 \label{eq42}
\end{align}
for all $l$ and $m$. Note that relations (\ref{eq64})-(\ref{eq24})
hold since $|\lambda_l(A)|\leq1$ for all $l$ in view of the fact
that $A$ is left-stochastic and, hence, $\rho(A)=1$. We therefore
established (\ref{eq92}). On the other hand, relation (\ref{eq42})
would hold if, and only if,
\begin{equation}\label{eq43}
    |\lambda_l(A)|^2\left[1+(1-\mu\lambda_m(R_u))^2
    -|\lambda_l(A)-\mu\lambda_m(R_u)|^2\right]\leq 1.
\end{equation}
To establish that (\ref{eq43}) is true for all $l$ and $m$, we
introduce the compact notation $\lambda=\lambda_l(A)$,
$\delta=\mu\lambda_m(R_u)$, and consider the following function of
two variables:
\begin{equation}\label{eq30}
    f(\lambda,\delta) \triangleq |\lambda|^2\left[1+(1-\delta)^2-|\lambda-\delta|^2\right]
    \text{ with $|\lambda|\leq 1$, $\delta\in(0,2)$, and $|\lambda-\delta|<1$.}
\end{equation}
The range for $\delta$ ensures condition (\ref{eq57}) and the
stability of the diffusion network, while the range for
$|\lambda-\delta|$ ensures that the consensus network is stable,
i.e., $|\lambda_{l,m}(\mathcal{B}_\text{cons})|<1$ for all $l$ and
$m$. Then, we would like to show that $f(\lambda,\delta)\leq 1$.
Since $\lambda$ is generally complex-valued, we denote the real part
of $\lambda$ by $\lambda_r$. Then, the term $|\lambda-\delta|^2$ in
(\ref{eq30}) is given by $|\lambda-\delta|^2=
|\lambda|^2+\delta^2-2\lambda_r\delta$ and $f(\lambda,\delta)$ from
(\ref{eq30}) becomes
\begin{equation}\label{eq32}
\begin{aligned}
    f(\lambda,\delta) =
    -|\lambda|^4+2(1-\delta+\lambda_r\delta)|\lambda|^2.
\end{aligned}
\end{equation}
Since $f(\lambda,\delta)$ is linear in $\delta$, the maximum value
of $f(\lambda,\delta)$ in (\ref{eq32}) over $\delta$ occurs at the
end points of $\delta$. Since $\delta\in(0,2)$ and
$|\lambda_r-\delta|\leq |\lambda-\delta|<1$, we conclude that
$0<\delta<1+\lambda_r$. Substituting the end points of $\delta$ into
(\ref{eq32}), we have
\begin{align}
    f(\lambda,0) &= -(|\lambda|^2-1)^2+1 \leq 1\\
    f(\lambda,1+\lambda_r) &= -|\lambda|^4+2\lambda_r^2|\lambda|^2
    \leq |\lambda|^4\leq 1
\end{align}
where we used the fact that $\lambda_r^2\leq |\lambda|^2$ and
$|\lambda|\leq 1$. We therefore established (\ref{eq93}).

Let us now examine what happens when the step-size is such that
$1\leq\mu \lambda_{\min}(R_u)<2$. Again, from (\ref{eq39}) and Table
III, we establish that ${\rm{MSD}}_{\rm{ncop}} \leq
{\rm{MSD}}_{\rm{cons}}$ this conclusion by showing that the ratio of
the individual terms appearing in the sums (\ref{eq39}) is upper
bounded by one:
\begin{equation}\label{eq122}
    \frac{1-|\lambda_l(A)-\mu\lambda_m(R_u)|^2}
    {1-(1-\mu\lambda_m(R_u))^2}\leq 1
\end{equation}
for all $l$ and $m$. Condition (\ref{eq122}) is equivalent to
showing that
\begin{equation}\label{eq69}
\begin{aligned}
    |\lambda_l(A)-\mu\lambda_m(R_u)|^2-(1-\mu\lambda_m(R_u))^2
    =|\lambda|^2-2\lambda_r\delta - (1-2\delta) \geq 0
\end{aligned}
\end{equation}
where we used the notation from (\ref{eq30}). Relation (\ref{eq69})
holds since $\delta\geq\mu\lambda_{\min}(R_u)\geq1\geq|\lambda|\geq
|\lambda_r|$ and then
\begin{equation}
\begin{aligned}
    |\lambda|^2-2\lambda_r\delta - (1-2\delta)\geq
    \lambda_r^2-2\lambda_r\delta - (1-2\delta)
    = (1-\lambda_r)(2\delta-1-\lambda_r)
    \geq 0.
\end{aligned}
\end{equation}

\section{Proof of Theorem \ref{thm_5}} \label{app_E}
From the series forms of $\{\text{MSD}_k(m)\}$ in Table IV, the
difference between MSD$_{\text{cta},k}(m)$ and
MSD$_{\text{ncop},k}(m)$ can be expressed as:
\begin{equation}\label{eq71}
\begin{aligned}
    {\rm{MSD}}_{\text{ncop},k}(m)-{\rm{MSD}}_{\text{cta},k}(m)=
    \mu^2\lambda_m(R_u)\sum_{j=0}^\infty(1-\mu\lambda_m(R_u))^{2j}e^T_k
    \left(\Sigma_v-A^{Tj}\Sigma_vA^{j}\right)e_k.
\end{aligned}
\end{equation}
From (\ref{eq80}), we have
\begin{equation}
    \lim_{j\rightarrow \infty}e^T_k\left(\Sigma_v-A^{Tj}\Sigma_vA^{j}\right)e_k
    = \sigma^2_{v,k}-e^T_kr_1s^T_1\Sigma_v s_1r_1^Te_k.
\end{equation}
Therefore, there exists an integer $J_m$ such that for any
$\varepsilon
> 0$,
\begin{equation}\label{eq70}
    e^T_k \left(\Sigma_v-A^{Tj}\Sigma_vA^{j}\right)e_k
    \geq \sigma^2_{v,k}-e^T_kr_1s^T_1\Sigma_v s_1r_1^Te_k-\varepsilon
    \triangleq \Delta
\end{equation}
for all $j\geq J_m$. From (\ref{eq68}), $\Delta$ in (\ref{eq70})
becomes $\Delta = \sigma^2_{v,k}-s^T_1\Sigma_v s_1/N-\varepsilon$.
From condition (\ref{eq103}), we are able to choose $\varepsilon$
small enough such that $\Delta$ is strictly greater than zero.
Therefore, expression (\ref{eq71}) is lower bounded by:
\begin{equation}\label{eq130}
\begin{aligned}
    {\rm{MSD}}_{\text{ncop},k}(m)-{\rm{MSD}}_{\text{cta},k}(m)\geq
    \mu^2\lambda_m(R_u)\left[-z+\Delta\cdot\sum_{j=J_m}^\infty
    (1-\mu\lambda_m(R_u))^{2j}\right]
\end{aligned}
\end{equation}
where the term $z\geq0$ is an upper bound for the first $J_m$ terms
of the summation in (\ref{eq71}), i.e.,
\begin{equation}
    \left|\sum_{j=0}^{J_m-1}(1-\mu\lambda_m(R_u))^{2j}e^T_k
    \left(\Sigma_v-A^{Tj}\Sigma_vA^{j}\right)e_k\right|\leq
    z<\infty.
\end{equation}
It can be verified that the series inside the brackets of
(\ref{eq130}) is strictly decreasing in
$\mu\in(0,1/\lambda_m(R_u))$. In addition,
\begin{equation}
    \lim_{\mu\rightarrow 0}\left(\sum_{j=J_m}^\infty
    (1-\mu\lambda_m(R_u))^{2j}\right) = \infty.
\end{equation}
Thus, there exists a $\mu^\circ_m > 0$ such that the sum inside the
bracket of (\ref{eq130}) becomes positive and, hence,
\begin{equation}\label{eq72}
    {\rm{MSD}}_{\text{ncop},k}(m)-{\rm{MSD}}_{\text{cta},k}(m)> 0
\end{equation}
for all $0 < \mu \leq \mu^\circ_m$. Repeating the above argument, we
will obtain a collection of step-size bounds
$\{\mu^\circ_1,\mu^\circ_2,\ldots,\mu^\circ_M\}$. We then choose
$\mu^\circ=\min\{\mu^\circ_1,\mu^\circ_2,\ldots,\mu^\circ_M\}$ so
that relation (\ref{eq72}) holds for all $m$. Then, applying Lemma
\ref{lem_4}, we arrive at (\ref{eq104}) for any $\mu$ satisfying
$0<\mu\leq \mu^\circ$.

\section{Condition (\ref{eq101}) Implies Condition (\ref{eq103}) when $A$ is
Primitive}\label{app_F} It follows from (\ref{eq101}) that
$A^{Tj}\Sigma_vA^j-A^{T(j+1)}\Sigma_vA^{j+1}\geq 0$ for any
nonnegative integer $j$ and then
\begin{equation}
    \sum_{j=0}^J\left(A^{Tj}\Sigma_vA^j-A^{T(j+1)}\Sigma_vA^{j+1}\right)
    =\Sigma_v-A^{T(J+1)}\Sigma_vA^{J+1}\geq 0.
\end{equation}
Since $A$ is primitive, as $J$ tends to infinity, we get from
(\ref{eq80}) that
\begin{equation}\label{eq105}
    \lim_{J\rightarrow \infty}\left(\Sigma_v-A^{T(J+1)}\Sigma_vA^{J+1}\right)
    =\Sigma_v-r_1s_1^T\Sigma_vs_1r_1^T \geq 0.
\end{equation}
Using (\ref{eq68}), we conclude that
\begin{equation}\label{eq48}
    \text{det}(\Sigma_v-r_1s_1^T\Sigma_vs_1r_1^T)=
    \text{det}(\Sigma_v)\cdot\text{det}\left(I_N-\Sigma_v^{-1}\mathds{1}\cdot
    \frac{s_1^T\Sigma_vs_1}{N}\mathds{1}^T\right)
    \geq 0.
\end{equation}
Since for any column vectors $\{x,y\}$ of size $N$, it holds that
$\text{det}(I_N-x\cdot y^T)=1-y^T\cdot x$, relation (\ref{eq48})
implies that the following must hold:
\begin{equation}\label{eq76}
    \left(1-\frac{s_1^T\Sigma_vs_1}{N}\mathds{1}^T\cdot\Sigma_v^{-1}\mathds{1}\right)
    \geq 0.
\end{equation}
However, by the Cauchy-Schwarz inequality \cite{Horn85} and using
the fact that $s_1^T\mathds{1}/\sqrt{N}=1$ from (\ref{eq68}), we
have
\begin{equation}\label{eq79}
\begin{aligned}
    \frac{s_1^T\Sigma_vs_1}{N}\mathds{1}^T\cdot\Sigma_v^{-1}\mathds{1}
    =\left(\sum_{l=1}^N\sigma^2_{v,l}\frac{s_{l,1}^2}{N}\right)
    \cdot\left(\sum_{l=1}^N\sigma^{-2}_{v,l}\right) \geq
    \left(\sum_{l=1}^N\frac{s_{l,1}}{\sqrt{N}}\right)^2
    =\left(\frac{s_1^T\mathds{1}}{\sqrt{N}}\right)^2=1
\end{aligned}
\end{equation}
where $s_{l,1}$ denotes the $l$th entry of $s_1$. Therefore,
relation (\ref{eq76}) can hold only with equality in (\ref{eq79}).
In turn, equality in (\ref{eq79}) holds if, and only if, there
exists a constant $c$ such that $s_{l,1}/\sqrt{N}=c\cdot
\sigma^{-2}_{v,l}$ for all $l$. By the fact that
$s_1^T\mathds{1}/\sqrt{N}=1$, we get:
\begin{equation}
    \frac{s_{l,1}}{\sqrt{N}} = \frac{\sigma^{-2}_{v,l}}
    {\sum_{m=1}^N\sigma^{-2}_{v,m}}
\end{equation}
and arrive at (\ref{eq103}) since
\begin{equation}
\begin{aligned}
    \sigma^2_{v,k}-\frac{s_1^T\Sigma_v s_1}{N}
    = \sigma^2_{v,k}-\frac{1}{\sum_{l=1}^N\sigma^{-2}_{v,l}}
    > \sigma^2_{v,k}-\frac{1}{\sigma^{-2}_{v,k}}
    =0.
\end{aligned}
\end{equation}

\bibliographystyle{IEEEtran}
\bibliography{IEEEfull,refs}

\end{document}